\documentclass[lettersize, journal]{IEEEtran}
\usepackage{amsmath,amsfonts}
\usepackage{array}
\usepackage[caption=false,font=normalsize,labelfont=sf,textfont=sf]{subfig}
\usepackage{textcomp}
\usepackage{stfloats}
\usepackage{url}
\usepackage{verbatim}
\usepackage{graphicx}
\usepackage{cite}
\usepackage{soul}
\usepackage{url}
\usepackage{ifthen}
\usepackage{cite}
\usepackage{braket}
\usepackage{amssymb}
\usepackage{amsthm}
\usepackage{graphicx}
\usepackage{dsfont}
\usepackage{tikz}
\usepackage{algorithm}
\usepackage{algpseudocode}
\usepackage{comment}

\newtheorem{theorem}{Theorem}
\newtheorem{definition}{Definition}
\newtheorem{proposition}{Proposition}

\newtheorem{remark}{Remark}
\newtheorem{claim}{Claim}

\newtheorem{corollary}{Corollary}
\newtheorem{example}{Example}

\DeclareMathOperator{\Log}{Log}

\begin{document}
\title{Codeword Stabilized Codes \\ from $m$-Uniform Graph States}
\author{Sowrabh Sudevan, Sourin Das, Thamadathil Aswanth, Nupur Patanker, Navin Kashyap, \IEEEmembership{Senior Member, IEEE}
\thanks{This paper was presented in part as a poster at the Quantum Information Knowledge (QuIK) workshop, held as a part of the 2024 IEEE International Symposium on Information Theory (ISIT 2024), Athens, Greece, July 7, 2024.}
\thanks{Sowrabh Sudevan and Sourin Das are with the Department of Physical Sciences, Indian Institute of Science Education and Research, Kolkata, 741246, India (email: \{ss18ip003,sourin\}@iiserkol.ac.in). Thamadathil Aswanth, Nupur Patanker and Navin Kashyap are with the Department of Electrical Communication Engineering, Indian Institute of Science, Bengaluru, 560012, India (email: \{aswantht,nupurp,nkashyap\}@iisc.ac.in).}}
\maketitle

\begin{abstract} 
 An $m$-uniform quantum state on $n$ qubits is an entangled state in which every $m$-qubit subsystem is maximally mixed. 
 Starting with an $m$-uniform state realized as the graph state associated with an $m$-regular graph, and a classical $[n,k,d \ge m+1]$ binary linear code with certain additional properties, we show that pure $[[n,k,m+1]]_2$ quantum error-correcting codes (QECCs) can be constructed within the codeword stabilized (CWS) code framework. As illustrations, we construct pure $[[2^{2r}-1,2^{2r}-2r-3,3]]_2$ and $[[(2^{4r}-1)^2, (2^{4r}-1)^2 - 32r-7, 5]]_2$ QECCs. We also give measurement-based protocols for encoding into code states and for recovery of logical qubits from code states. 
\end{abstract}

\section{Introduction}

A quantum computer exploits the properties of quantum mechanics to do computations and outperforms classical computers in solving certain problems \cite{Shor_1997,10.1145/237814.237866}. Any realistic physical realization of a quantum computer must account for noise \cite{doi:10.1126/science.270.5242.1633,Plenio_1997,1996,PhysRevA.51.992}. In quantum computing, errors or noise can be modeled as some extra, unknown unitaries stochastically introduced into the intended quantum circuit \cite{knill1999theory}. Much like for classical information, the effects of noise on quantum memories can be dealt with through error correction procedures, the first of which was designed by Shor \cite{PhysRevA.52.R2493}. The theory of stabilizer codes, developed by Gottesman \cite{gottesman1997stabilizer}, provided a general algebraic framework for constructing quantum error-correcting codes (QECCs). Codes constructed via the stabilizer formalism are generally referred to as ``additive'' codes. Following the development of stabilizer codes, much progress has been made in quantum error correction, and finding new QECCs has been an active area of research \cite{kitaev1997fault,Dennis_2002,Fowler_2012,Breuckmann_2021,Lai_2012,Grassl_1997,Grassl1999QuantumBC,barkeshli2023highergroup,Brun2006CorrectingQE,Wang2022,Gottesman_2001,PRXQuantum.2.020101,Zhang_2021}. Furthermore, different methods to construct non-additive quantum codes have been proposed (see \cite{GR-chapter-in-LB}); examples of such codes include union stabilizer codes \cite{Grassl1,grassl2} and codeword stabilized (CWS) codes \cite{CWSCode, CWSCode2}. 

The CWS code construction provides a versatile framework within which good quantum codes --- additive as well as non-additive --- can be obtained \cite{CWSCode},\cite{GR-chapter-in-LB}. The building blocks of the construction\footnote{In this paper, we restrict our attention to quantum codes on qubits, but the CWS code construction has also been extended to larger qudit alphabets \cite{CWSCode2}.} are a graph, which determines a graph state, and a classical binary code, which may be non-linear. The choices made for the graph and the classical code determine the properties, such as the minimum distance, of the resulting CWS code. However, to the best of our knowledge, the prior literature does not provide a principled approach for choosing the graph and classical code in such a way as to yield a quantum code with a guaranteed minimum distance. Indeed, the original work of Cross et al.\ \cite{CWSCode} leaves this as an open problem, while subsequent works have considered exhaustive search strategies \cite{GR-chapter-in-LB}.

In this paper, we give a prescription for choosing a graph and a classical binary \emph{linear} code to obtain an additive CWS code with a guaranteed minimum distance. All codes we obtain using our prescription are \emph{pure}\footnote{As defined in \cite{Calderbank1998}, a pure QECC is one in which distinct elements of the set of correctable errors produce orthogonal results.} QECCs.
The initial quantum resource for our construction is an $m$-uniform quantum state, which is a multipartite entangled state within which any $m$-qubit subsystem is maximally mixed \cite{PhysRevA.87.012319}. Such a state spans a pure QECC of dimension one and minimum distance $m+1$ \cite{PhysRevA.69.052330,Huber_2020}. Since a $1$-dimensional QECCs cannot by itself store any logical information, some mechanism is needed to enlarge the code space so that it can hold a positive number of logical qubits. The framework of CWS codes provides one such mechanism, although as we shall see in Section~\ref{sec:related_work} below, other mechanisms have also been proposed in the literature. 

In the CWS code construction, we start with a graph $G$, which yields a graph state $\ket{\mathrm{G}}$, and form a quantum code, called a CWS code, whose basis is given by the collection of states resulting from the action of certain strings of Pauli-$Z$ operators on $\ket{\mathrm{G}}$. The Pauli-$Z$ strings in the construction are determined by the codewords of a classical binary code. The CWS code thus obtained is additive if and only if the classical binary code is linear. The number of logical qubits encoded by the CWS code is equal to the dimension of the binary linear code. 

\subsection{Our Contributions} \label{sec:contrib}
In this paper, given an $m$-uniform graph state $\ket{\mathrm{G}}$, we give a necessary and sufficient condition (Theorem~\ref{theorem_1}) on a binary linear code $\mathcal{C}$ for the CWS code constructed from $\ket{\mathrm{G}}$ and $\mathcal{C}$ to be a pure QECC with minimum distance $m+1$. As an application, we show that if the $m$-uniform graph state $\ket{\mathrm{G}}$ comes from an $m$-regular graph, then any binary linear code $\mathcal{C}$ of dimension $k$ and minimum distance $d > m(m+1)$ will suffice to ensure that the resulting CWS code is a pure QECC with minimum distance $m+1$ that can encode $k$ logical qubits. We apply this result to the graph state associated with a specific $2D$-regular graph, namely, the $D$-dimensional regular lattice with periodic boundary conditions; the associated graph state is called a $D$-dimensional cluster state. Sudevan et al.\ \cite{sowrabh} have shown that a $D$-dimensional cluster state is $2D$-uniform. Thus, from $D$-dimensional cluster states, and binary linear codes of dimension $k$ and minimum distance strictly greater than $2D(2D+1)$, we obtain, via the CWS code framework, explicit constructions of pure QECCs with $k$ logical qubits and minimum distance $2D+1$. For $1$- and $2$-dimensional cluster states, we give explicit, customized constructions of binary linear codes that yield QECCs that can hold a larger number of logical qubits than those obtained from binary linear codes with minimum distance $d > 2D(2D+1)$.

As a further contribution, we observe that additive CWS codes allow for measurement-based approaches to encoding of logical qubits into the code space and recovery of logical qubits from the encoded state in the code space. The encoding procedure uses some ideas from measurement-based quantum computing (MBQC) \cite{Briegel2009,NIELSEN2006147}. MBQC started with the idea of a one-way quantum computer, proposed by Raussendorf et al.\ \cite{oneWayQC}, who showed that any quantum operation can be realized through single qubit measurements on a cluster state. The key idea we borrow from MBQC is a method called one-bit teleportation \cite{PhysRevA.62.052316}. Our encoding procedure teleports information from input qubits with arbitrary quantum information to the CWS code space using a combination of controlled-$Z$ gates and $X$ measurements. In this procedure, one can encode the logical qubits sequentially: when a subset of the logical qubits is provided, the encoder maps them to a code state in a subcode (subspace) of the full code space. This code state is further transformed into a state of the full code when the remaining logical qubits are encoded. This property requires the logical qubits initially encoded to be unentangled with the remaining qubits. Our protocol for retrieving the logical qubits from an encoded state is also measurement-based. The protocol allows for partial recovery of logical qubits: the initial subset of logical qubits can be recovered from the final encoded state, leaving the remaining logical qubits still encoded. 

The sequential encoding and partial recovery of logical qubits are interesting and potentially useful properties of our measurement-based protocols.
These protocols also give us cause to assign the name ``tent peg code'' to the binary linear code used in the CWS construction. This nomenclature is motivated by the $\text{controlled-}Z$ gates in the encoding circuit that appear like lines of tents (see Figure~\ref{kAncillas}).

\subsection{Prior Work on QECCs from $m$-Uniform States}\label{sec:related_work}
Our work builds upon that of Sudevan et al.\ \cite{sowrabh}, who proposed a method for constructing $2\text{-dimensional}$ QECCs using $2D$-uniform cluster states, that can be encoded in a measurement-based approach. In our paper, we extend this idea to higher-dimensional codes that can encode multiple logical qubits. 

A quantum code construction in the prior literature that is perhaps closest in spirit to ours is that of Kapshikar and Kundu \cite{HardnessOfQuantumCodes}. Their work uses the CWS code framework to show the NP-hardness of the problem of determining the minimum distance of a quantum code. This is done via a reduction from the same problem for a classical binary linear code, which is known to be NP-hard. The graph states used in their construction are obtained from graphs that have no $4$-cycles. They show that if a graph has minimum degree at least $m$ and no $4$-cycles, then the associated graph state is $m$-uniform. Their NP-hardness result is obtained by constructing a $[[\Theta(n^2),k,d]]$ quantum code from a given $[n,k,d]$ classical binary linear code. The construction applies the CWS code framework on graph states obtained from an explicit family of $4$-cycle-free graphs due to Erd\"os, R\'enyi and S\'os \cite{ERS66}. Our work differs fundamentally from \cite{HardnessOfQuantumCodes} in that the graph states used in our construction are cluster states obtained from graphs (lattices) that have an abundance of $4$-cycles.

There are other methods in the literature that use $m$-uniform states (which need not be graph states) for constructing quantum codes that can encode a positive number of logical qudits of local dimension $q \ge 2$. Most of these constructions (e.g., \cite{gottesman1997stabilizer,Raissi2020,Huber_2020,Alsina_2021}) use the idea of shortening or partially tracing out one qudit of a pure QECC \cite{Rains1998}. In particular, Raissi \cite{Raissi2020} showed that starting from an $m$-uniform stabilizer state, i.e., a pure $[[n,0,m+1]]_q$ stabilizer code, repeatedly applying the shortening procedure yields pure $[[n-j,j,m+1-j]]_q$ QECCs, for $j=1,2,\ldots,m-1$. Note that the QECCs thus obtained have smaller blocklength and minimum distance. This is in contrast to our CWS code construction, which retains the same blocklength and minimum distance as the pure QECC corresponding to the graph state that we start with.

Raissi \cite{Raissi2020} also gave a ``modified-shortening'' construction that does not require tracing out any qudits from a parent QECC. This construction starts with an absolutely maximally entangled (AME) state, i.e., an $\lfloor n/2 \rfloor$-uniform state on $n$ qudits of local dimension $q$, or equivalently, an $[[n,0,\lfloor n/2 \rfloor+1]]_q$ QECC, and produces an $[[n,1,\lfloor n/2 \rfloor]]_q$ QECC. The AME state is obtained from a classical $[n,\lfloor n/2 \rfloor,\lceil n/2 \rceil + 1]_q$ maximum distance separable (MDS) code. Raissi's modified-shortening construction is not directly comparable to our CWS code construction, since our construction yields QECCs over qubits ($q=2$), whereas Raissi's construction requires $q$ large enough to ensure the existence of the required $[n,\lfloor n/2 \rfloor,\lceil n/2 \rceil + 1]_q$ MDS code.

\subsection{Organization} The remainder of this paper is organized as follows. Section~\ref{sec:prelims} introduces the basic definitions and notation used in this paper. Section~\ref{sec:main} presents our main results, including explicit constructions of families of quantum codes obtained from $1$- and $2$-dimensional cluster states. The measurement-based protocols for encoding and recovery of logical qubits are discussed in Section~\ref{sec: encoding}. 
Proofs omitted from the main body of the paper are given in the appendices.


\medskip
\section{Preliminaries} \label{sec:prelims}
Throughout this paper, for a positive integer $n$, we set $[n] := \{1,2, \ldots, n\}$, and for integers $m \le n$, $[m:n] := \{m,m+1,\ldots,n\}$. For a binary vector $\mathbf{x} = [x_1,x_2,\ldots,x_n] \in \{0,1\}^n$, define its \emph{support} as $\text{supp}(\mathbf{x}) := \{i: x_i = 1\}$. The \emph{characteristic vector} of a subset $A \subseteq [n]$ is the binary vector $\mathbf{a} \in \{0,1\}^n$ such that $A = \text{supp}(\mathbf{a})$. The binary field formed by $\{0,1\}$ equipped with modulo-$2$ arithmetic is denoted by $\mathbb{F}_2$.

$\mathcal{P}_n$ denotes the Pauli group on $n$ qubits, and $X_i$ and $Z_i$ respectively denote the Pauli-$X$ and Pauli-$Z$ operators acting on $i^{th}$ qubit. A tensor product of Pauli-$Z$ operators is specified via a binary vector: for $\mathbf{x} = [x_1, x_2, \ldots, x_n] \in {\{0,1\}}^n$, we use $Z_\mathbf{x}$ to denote the Pauli string $\prod_{j =1}^{n}{Z_j^{x_j}}$. Alternatively, for a subset of qubits $A \subseteq [n]$, we use $Z_{A}$ to denote $\prod_{j\in A}Z_j$. Note that $Z_A = Z_\mathbf{x}$ if and only if $\mathbf{x} = [x_1, x_2, \ldots, x_n]$ is the characteristic vector of $A$. Analogous notation is used for tensor products of Pauli-X operators as well, whenever needed.

For a Pauli operator $M$ of the form (ignoring the global phase factor) $\sigma^{(1)}\otimes \sigma^{(2)}\otimes \ldots \otimes \sigma^{(n)}$, with $\sigma^{(j)}\in \{X,Y,Z,I\} \hspace{5pt} \forall j\in[n]$, we denote its weight by $\mathrm{wt}(M) = |\{i: \sigma^{(i)} \neq I\}|$. 
It is often convenient to use the symplectic notation for Pauli operators, where a Pauli operator $M \in \mathcal{P}_n$ is uniquely identified, up to a global phase factor, with a binary vector of the form $\mathbf{v} = [\mathbf{v}^{(1)}|\mathbf{v}^{(2)}]$, with $\mathbf{v}^{(1)},\mathbf{v}^{(2)} \in \{0,1\}^n$: 
\begin{equation*}
    M = (\text{global phase factor}) \cdot X_{\mathbf{v}^{(1)}} \cdot Z_{\mathbf{v}^{(2)}}.
\end{equation*}
In this representation, the multiplication of operators is equivalent, up to a global phase, to the coordinatewise modulo-$2$ addition of the corresponding binary vectors. When using this representation, the symplectic weight of $\mathbf{v} = [\mathbf{v}^{(1)}|\mathbf{v}^{(2)}]$ is a useful notion; it is defined as
\begin{equation*}
    \mathrm{wt}_\mathrm{s}(\mathbf{v}) = |\{j: (v^{(1)}_j,v^{(2)}_j) \neq (0,0)\}|,
\end{equation*}
where $\mathbf{v}^{(1)}=[v^{(1)}_1, v^{(1)}_2, \ldots, v^{(1)}_n]$ and $\mathbf{v}^{(2)} = [v^{(2)}_1, v^{(2)}_2, \ldots, v^{(2)}_n]$. This is equal to the weight of the Pauli operator corresponding to $\mathbf{v}$. 

An $[n,k]$ binary linear code has length $n$ and dimension $k$, i.e., it bijectively maps $k$ data bits into $n$ coded bits. If, in addition, it has minimum Hamming distance $d$, then it is an $[n,k,d]$ binary linear code. Analogously, an $[[n,k]]$ QECC\footnote{Since all our quantum codes are defined on qubit spaces, we drop the subscript $2$ from the notation $[[n,k]]_2$ and $[[n,k,d]]_2$.} bijectively maps the state of $k$ logical qubits to a subspace of the Hilbert space of $n$ physical qubits. Following \cite{Huber_2020}, we define a \emph{pure} $[[n,k,d]]$ quantum code to be a $2^k$-dimensional subspace, $\mathcal{Q}$, of ${(\mathbb{C}^2)}^{\otimes n}$ such that for all Pauli operators $E \in \mathcal{P}_n$ of weight strictly between $0$ and $d$, we have $\braket{\psi | E | \psi} = 0$ for all $\ket{\psi} \in \mathcal{Q}$. 

\subsection{Graph states and $m$-uniformity}

We start with the standard definition of a graph state.
\begin{definition} 
    Consider a set of $n$ qubits and an undirected, simple graph $G$ on the vertex set $[n]$. The graph state $\ket{\mathrm G}$ associated with $G$ is the simultaneous $+1$-eigenstate of the stabilizers generated by $\{S_i := X_i Z_{N(i)}: i \in [n]\}$, where $N(i)$ is the set of neighbours of vertex $i$ in the graph $G$.
\end{definition}

Throughout this paper, we will use $S_i$ to denote the stabilizer generator associated with the vertex $i$. Since the $S_i$'s are a set of $n$ independent generators on $n$ qubits, the graph state $\ket{\mathrm{G}}$ is a \emph{stabilizer state}, i.e., an $[[n,0,d]]$ stabilizer code for some $d \ge 1$.

In the literature, we find the terms graph states and cluster states being used interchangeably. In this work, we use \emph{graph states} to denote states associated with arbitrary undirected simple graphs, and \emph{cluster states} for graph states associated with regular lattices. Of particular significance to our work are cluster states obtained from $D$-dimensional rectangular lattices with periodic boundary conditions. For positive integers $D, n_1, n_2, \ldots, n_D$, let $\Lambda_{n_1,n_2,\ldots,n_D}$ denote the $D$-dimensional rectangular lattice with periodic boundary conditions and side lengths $n_1, n_2, \ldots, n_D$. Mathematically, this is the (undirected) graph on the vertex set $[0:n_1-1] \times [0:n_2-1] \times \cdots \times [0:n_D-1]$, with an edge drawn between vertices $\mathbf{x} = (x_1,x_2,\ldots,x_D)$ and $\mathbf{y} = (y_1,y_2,\ldots,y_D)$ iff the Lee distance
$$d_L(\mathbf{x},\mathbf{y}) := \sum_{i=1}^D \min(|x_i-y_i|, n_i - |x_i-y_i|)$$
between $\mathbf{x}$ and $\mathbf{y}$ equals $1$. As an illustration, Figure~\ref{fig:2DCluster} shows the lattice $\Lambda_{3,3}$.

\begin{figure}[h]
    \centering
    \includegraphics[width=0.45\linewidth]{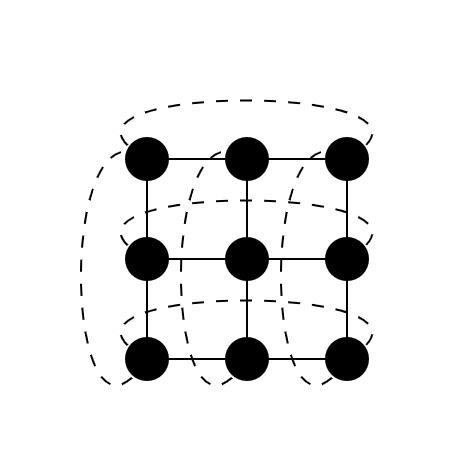}
    \caption{$\Lambda_{3,3}$: An example of $2$-dimensional rectangular lattice with periodic boundary conditions. The periodic boundaries are shown with dotted lines.}
    \label{fig:2DCluster}
\end{figure}

In general, graph states can be highly entangled. The type of multipartite entanglement that is relevant to us is $m$-uniformity.
\begin{definition} \label{def:m-unif}
    A quantum state $\ket{\psi}$ is said to be $m$-uniform if all possible $m$-qubit marginals are maximally mixed. Further, a quantum code is said to be $m$-uniform when every state it contains is $m$-uniform.
\end{definition}
An equivalent characterization of an $m$-uniform state is as follows: An $n$-qubit quantum state $\ket{\psi}$ is $m$-uniform iff $\bra{\psi}\mathcal{O}\ket{\psi}=0$ for all non-identity operators $\mathcal{O} \in \mathcal{P}_n$ with weight at most $m$ \cite{PhysRevA.87.012319}. Expressed in the language of quantum codes, an $m$-uniform state $\ket{\psi}$ spans a pure $[[n,0,m+1]]$ QECC. 
For graph states, there is a simple characterization of $m$-uniformity, stated in the proposition below.

\begin{proposition} \label{prop:degree}
 A graph state $\ket{\mathrm{G}}$ associated with a graph $G$ is $m$-uniform iff all the stabilizers in $\mathcal{S} = \langle S_1,S_2,\ldots,S_n \rangle$, other than the identity $I$, have weight at least $m+1$. In particular, if $\ket{\mathrm{G}}$ is $m$-uniform, then every vertex of $G$ has degree at least $m$.
\end{proposition}

\begin{proof}
    For any stabilizer $S$ of $\ket{\mathrm{G}}$, we have $\braket{\mathrm{G} | S | \mathrm{G}} = \braket{\mathrm{G} | \mathrm{G}} = 1 \ne 0$. Hence, if some stabilizer $S \ne I$ of $\ket{\mathrm{G}}$ has weight at most $m$, then $\ket{\mathrm{G}}$ cannot be $m$-uniform.

    Conversely, suppose that no stabilizer $S \ne I$ of $\ket{\mathrm{G}}$ has weight at most $m$. Consider any non-identity Pauli operator $\mathcal{O} \in \mathcal{P}_n$ of weight at most $m$. We then have $\mathcal{O} \notin \mathcal{S}$. Since $\ket{\mathrm{G}}$ is a stabilizer state, the centralizer, $\mathcal{C}(\mathcal{S})$, of $\mathcal{S}$ is the same as $\mathcal{S}$. Hence, $\mathcal{O} \notin \mathcal{C}(\mathcal{S})$, from which it follows that $\braket{\mathrm{G} | \mathcal{O} | \mathrm{G}} = 0$. Therefore, $\ket{\mathrm{G}}$ is $m$-uniform.

    Finally, if $\ket{\mathrm{G}}$ is $m$-uniform, then each $S_i$ has weight at least $m+1$. Since the weight of $S_i$ equals $1 + |N(i)|$, we must have $|N(i)|\geq m$, i.e., the degree of each vertex $i$ is at least $m$.
\end{proof}

The following useful result, shown by Sudevan et al.\ \cite[Appendix~C]{sowrabh}, demonstrates that cluster states associated with sufficiently large rectangular lattices have the kind of multipartite entanglement we seek.

\begin{proposition} \label{prop:sudevan_etal}
The cluster state associated with the $D$-dimensional rectangular lattice $\Lambda_{n_1,n_2,\ldots,n_D}$ is $2D$-uniform when $n_i \ge 8$ for all $i$.
\end{proposition}

\subsection{CWS codes from graph states}

Codeword stabilized (CWS) codes were introduced by Cross et al.\ \cite{CWSCode} as a means of extending the stabilizer code formalism to encompass non-additive QECCs. Our starting point is a specific construction (called the standard form in \cite{CWSCode}) of an additive CWS code, which we summarize in the proposition below. The proposition is a direct consequence of Theorems~1 and 5 in \cite{CWSCode}. 

\begin{proposition}
    \label{prop:CWScode}
    Consider an $n$-qubit graph state, $\ket{\mathrm{G}}$, whose stabilizers are constructed from a graph $G$. Let $\mathcal{C}$ be a binary linear code of length $n$ and dimension $k$. Then, 
    the CWS code defined as $\text{span}\bigl(\{Z_{\mathbf{c}}\ket{\mathrm{G}}: \mathbf{c} \in \mathcal{C}\}\bigr)$ is an $[[n,k]]$ stabilizer code.
\end{proposition}

\begin{example}
Let $\ket{\mathrm{G}}$ be the graph state associated with the $5$-cycle $C_5$, and let $\mathcal{C}$ be the $5$-bit repetition code $\{00000, 11111\}$. Then, $\text{span}\bigl( \{\ket{\mathrm{G}}, Z_{11111}\ket{\mathrm{G}}\} \bigr)$ is a $[[5,1]]$ stabilizer code, which can be shown to have minimum distance $3$.
\end{example}

The following proposition, proved in Appendix~\ref{appendix:codeStabilizers}, describes the stabilizer group of this CWS code.
\begin{proposition} \label{codeStabilizers}
    Consider a graph state $\ket{\mathrm{G}}$ with stabilizer group $\mathcal{S}$. Let $\{\mathbf{a}_1, \mathbf{a}_2, \ldots, \mathbf{a}_k\}$ be a basis for a binary linear code $\mathcal{C}$, and let $\mathbf{a}_0 = \mathbf{0}$ be the zero vector. Then, the stabilizer group $\mathcal{S}_\mathcal{C}$ of $\text{span}\bigl(\{Z_{\mathbf{c}}\ket{\mathrm{G}}: \mathbf{c} \in \mathcal{C}\}\bigr)$ is
    \begin{equation*}
        \mathcal{S}_\mathcal{C} = \bigcap_{i=0}^{k}{Z_{\mathbf{a}_i}\mathcal{S}Z_{\mathbf{a}_i}},
    \end{equation*}
    where $Z_{\mathbf{a}_i}\mathcal{S}Z_{\mathbf{a}_i} = \{Z_{\mathbf{a}_i} S Z_{\mathbf{a}_i}: S \in \mathcal{S}\}$.
\end{proposition}

\begin{example}
Consider the graph state $\ket{\mathrm{G}}$ associated with the path, $P_3$, on $3$ vertices. The stabilizer group, $\mathcal{S}$, of $\ket{\mathrm{G}}$ has generators $g_1 = X_1 Z_2$, $g_2 = Z_1 X_2 Z_3$ and $g_3 = Z_2 X_3$. Let $\mathcal{C} = \{000, 110\}$. Note that $Z_{000}\mathcal{S}Z_{000} = \mathcal{S} = \{I, g_1, g_2, g_3, g_1 g_2 , g_2 g_3, g_1 g_3, g_1 g_2 g_3\} $, and $Z_{110}\mathcal{S}Z_{110} = (Z_1Z_2)\mathcal{S}(Z_1Z_2) = \{I, -g_1, -g_2, g_3, g_1 g_2, -g_2 g_3, -g_1 g_3, g_1 g_2 g_3\}$. Thus, from Proposition~\ref{codeStabilizers}, we have $\mathcal{S}_{\mathcal{C}} = \mathcal{S} \bigcap (Z_1 Z_2) \mathcal{S}(Z_1 Z_2) = \{I, g_1 g_2, g_3, g_1 g_2 g_3\}.$
\end{example}

Cross et al.\ \cite[Section~V]{CWSCode} give an encoding circuit for the CWS code that has the same complexity as the encoding circuit used for the classical linear code $\mathcal{C}$. In Section~\ref{sec: encoding} of our paper, we describe an alternative, measurement-based procedure for encoding. Owing to the role played by the code $\mathcal{C}$ in our encoding procedure (see Figure~\ref{kAncillas}), we hereafter refer to $\mathcal{C}$ as the \emph{tent peg code} in the CWS code construction.

\section{Main Results}\label{sec:main}

The basis of our construction of CWS codes with good error-correction properties is the following theorem, which shows that the graph state and the tent peg code need to be chosen carefully to obtain a good QECC.

\begin{theorem} \label{theorem_1}
    Let $\mathcal{S}$ denote the stabilizer group corresponding to a graph state $\ket{\mathrm{G}}$ on $n$ qubits, and let $\mathcal{V}_\mathcal{S}$ be the corresponding binary vector space in the symplectic representation. Let $\mathcal{C}$ be a binary linear code of length $n$ and dimension $k$, and consider its length-$2n$ extension $\mathcal{C}^\prime := \{[0,0, \ldots, 0 \ | \ \mathbf{c}]: \mathbf{c}\in \mathcal{C} \}$. Define $\mathcal{W}:= \mathcal{V}_\mathcal{S}+\mathcal{C}^\prime = \{\mathbf{v}+\mathbf{c}^\prime: \mathbf{v} \in \mathcal{V}_\mathcal{S} \text{ and } \mathbf{c}^\prime \in \mathcal{C}^\prime\}$, where the vector addition in $\mathbf{v} + \mathbf{c}'$ is coordinatewise addition modulo $2$. Then, $\mathcal{Q} := \text{span}\bigl(\{ Z_{\mathbf{c}}\ket{\mathrm{G}}: \mathbf{c} \in \mathcal{C} \}\bigr)$ is a pure $[[n,k,m+1]]$ QECC if and only if the minimum symplectic weight of a non-zero vector in $\mathcal{W}$ is at least $m+1$.
\end{theorem}

\begin{proof}
Suppose that there is no non-zero vector of symplectic weight at most $m$ in $\mathcal{W}$. Then, for any non-identity operator $E \in \mathcal{P}_n$ of weight at most $m$, we must have $Z_\mathbf{c}E \notin \mathcal{S}$ for all $\mathbf{c} \in \mathcal{C}$. This implies
    \begin{equation}
        \label{eq:theorem1}
        \bra{\mathrm{G}}Z_\mathbf{c}E\ket{\mathrm{G}} = 0 \hspace{0.5cm}\forall \mathbf{c} \in \mathcal{C},
    \end{equation}
so that, for any $\mathbf{c}_1,\mathbf{c}_2 \in \mathcal{C}$, we have 
    \begin{equation*}
        \bra{\mathrm{G}}Z_{\mathbf{c}_1} Z_{\mathbf{c}_2} E\ket{\mathrm{G}} = 0.
    \end{equation*}
Since $Z_{\mathbf{c}_2} E = \pm E Z_{\mathbf{c}_2}$ and $Z_{\mathbf{c}_1}^\dagger = Z_{\mathbf{c}_1}$, we in fact have 
    \begin{equation}
    \label{eq:KL}
        \bra{\mathrm{G}}Z_{\mathbf{c}_1}^\dagger E Z_{\mathbf{c}_2}\ket{\mathrm{G}} =0 \hspace{0.5cm}\forall \mathbf{c}_1, \mathbf{c}_2 \in \mathcal{C}.
    \end{equation}
    Therefore, $\mathcal{Q} := \text{span}\{ Z_{\mathbf{c}}\ket{\mathrm{G}}: \mathbf{c} \in \mathcal{C} \}$ is a pure $[[n,k,m+1]]$ QECC. 

    Conversely, if $\mathcal{Q}$ is a pure $[[n,k,m+1]]$ QECC, then for any non-identity operator $E \in \mathcal{P}_n$ of weight at most $m$, Eq.~\eqref{eq:KL} must hold. This implies that Eq.~\eqref{eq:theorem1} also holds, from which it follows that the minimum symplectic weight of a non-zero vector in $\mathcal{W}$ is at least $m+1$.
 \end{proof}
 
Theorem~\ref{theorem_1} does not explicitly impose an $m$-uniformity condition on the graph state $\ket{\mathrm{G}}$. However, to obtain a pure $[[n,k,m+1]]$ QECC via the theorem, $\ket{\mathrm{G}}$ must necessarily be $m$-uniform. This is because the minimum symplectic weight of a non-zero vector in $\mathcal{W}$ can be at least $m+1$ only if $\ket{\mathrm{G}}$ is $m$-uniform. Indeed, $\mathcal{W}$ contains all the non-zero vectors in $\mathcal{V}_{\mathcal{S}}$, which correspond to the non-identity stabilizers of $\ket{\mathrm{G}}$. Therefore, the condition on the minimum symplectic weight of a non-zero vector in $\mathcal{W}$ implies that all the non-identity stabilizers of $\ket{\mathrm{G}}$ have weight at least $m+1$, so that, by Proposition~\ref{prop:degree}, $\ket{\mathrm{G}}$ must be $m$-uniform. Thus, in all the constructions that follow of $[[n,k,m+1]]$ QECCs based on Theorem~\ref{theorem_1}, we will always start with an $m$-uniform graph state $\ket{\mathrm{G}}$.

Given $n$ and an $m$-uniform graph state $\ket{\mathrm{G}}$, one can ask what the largest value of $k$ is for which an $[[n,k,m+1]]$ QECC can be constructed using Theorem~\ref{theorem_1}. However, this is not easy to determine. Consider the set, $\mathcal{B}_m(\mathcal{V}_\mathcal{S}) \subseteq \{0,1\}^{2n}$, consisting of all vectors of the form $\mathbf{v} + \mathbf{x}$ with $\mathbf{v} \in \mathcal{V}_\mathcal{S}$ and $\mathbf{x} \in \{0,1\}^{2n}$, $\mathrm{wt}_{\mathrm{s}}(\mathbf{x}) \le m$. (Here, again, the vector addition is coordinatewise, modulo $2$.) Let $\mathcal{V}_X$ denote the binary vectors corresponding to the Pauli operators containing at least one $X$. In order to get the tent peg code that maximizes the dimension of the resulting QECC, we need to find the largest subspace of $\mathbb{F}_2^{2n}$ that has no intersection with $\mathcal{B}_m(\mathcal{V}_{\mathcal{S}}) \cup \mathcal{V}_X$ except $\mathbf{0}$. This problem reduces to the 4-colorability problem, which is known to be NP-hard \cite{andreasBlass}. So, instead of trying to find the QECC of largest dimension that can be constructed in this way, we aim to identify some simple conditions on the linear code $\mathcal C$ that will result in QECCs of reasonably large dimension. 

As we saw in Proposition~\ref{prop:degree}, $m$-uniformity of $\ket{\mathrm{G}}$ implies that the degree of every vertex in $G$ is at least $m$. If we further insist that $\ket{\mathrm{G}}$ arises from a graph $G$ in which every vertex has degree equal to $m$ (i.e., an $m$-regular graph), then by choosing a tent peg code $\mathcal{C}$ of sufficiently large minimum distance, we obtain a construction of a pure $[[n,\dim\mathcal{C},m+1]]$ QECC via Theorem~\ref{theorem_1}. 

\begin{theorem}\label{theorem: distanceBased}
   Let $\ket{\mathrm{G}}$ be an $m$-uniform graph state {on $n$ qubits} constructed from an $m$-regular graph {with $n$ vertices}. Then, for any $[n,k,d]$ binary linear code $\mathcal{C}$ with $d>m(m+1)$, the quantum code $\mathcal{Q} := \text{span}\bigl(\{Z_{\mathbf{c}}\ket{\mathrm{G}}: \mathbf{c} \in \mathcal{C}\}\bigr)$ is a pure $[[n,k,m+1]]$ QECC.
\end{theorem}

\begin{proof}
    The idea is to show that in the vector space $\mathcal W$ defined in the statement of Theorem~\ref{theorem_1}, there are no non-zero vectors of symplectic weight at most $m$, so that we can draw the desired conclusion from Theorem~\ref{theorem_1}. We will show the equivalent statement that for any non-identity operator $E \in \mathcal{P}_n$ of weight at most $m$, the operator $Z_\mathbf{c}E$ is not a stabilizer of $\ket{\mathrm{G}}$, for all $\mathbf{c} \in \mathcal{C}$. To this end, consider any non-identity operator $E \in \mathcal{P}_n$ of weight at most $m$. Note first that $E$ itself cannot be a stabilizer of $\ket{\mathrm{G}}$, since by the characterization of $m$-uniformity given after Definition~\ref{def:m-unif}, we have $\braket{\mathrm{G} | E | \mathrm{G}} = 0$. 
    
    Now, consider $E'= Z_\mathbf{c}E$, for some $\mathbf{c} \in \mathcal{C}$, $\mathbf{c} \ne \mathbf{0}$. Since $\mathbf{c}$ has Hamming weight strictly greater than $m(m+1)$ and the weight of $E$ is at most $m$, the number of $Z$s in $E'$ must be strictly greater than $m^2$. On the other hand, the number of $X$s in $E'$ is at most $m$, since all of these have to be contributed by $E$. Suppose that $E'$ were a stabilizer of $\ket{\mathrm{G}}$. Then, it is a product of stabilizer generators of the form $S_i = X_i Z_{N(i)}$. Since the number of $X$s in $E'$ is at most $m$, $E'$ is a product of at most $m$ generators $S_i$. Moreover, each generator $S_i$ contributes at most $|N(i)| = m$ $Z$s to $E'$. So, the total number of $Z$s in $E'$ cannot exceed $m \times m = m^2$, which contradicts the lower bound derived earlier for the number of $Z$s in $E'$. Hence, $Z_\mathbf{c}E$ cannot be a stabilizer of $\ket{\mathrm{G}}$, which proves the theorem.
\end{proof}

Since Proposition~\ref{prop:sudevan_etal} gives us a concrete construction of a graph state that satisfies the requirement of the above theorem, we have the following corollary.

\begin{corollary}\label{cor:cluster_CWS}
    Let $\ket{\mathrm{G}}$ be the cluster state associated with the $D$-dimensional rectangular lattice $\Lambda_{n_1,n_2,\ldots,n_D}$, with $n_i \ge 8$ for all $i$. {Let $n = n_1 \times n_2 \times \cdots \times n_D$.} Then, for any $[n,k,d]$ binary linear code $\mathcal{C}$ with ${d > 2D(2D+1)}$, the quantum code $\mathcal{Q} := \text{span}\bigl(\{ Z_{\mathbf{c}}\ket{\mathrm{G}}: \mathbf{c} \in \mathcal{C} \}\bigr)$ is a pure $[[n,k,2D+1]]$ QECC.
\end{corollary}

{From the Gilbert-Varshamov bound (see e.g.,\cite[Section~4.3]{Roth2006}), we know that for any fixed $D \ge 1$, there exist $[n,k,d]$ binary linear codes with $d > 2D(2D+1)$ and $k/n \to 1$ as $n \to \infty$. More precisely, we have the following additional corollary to Theorem~\ref{theorem: distanceBased}. 
}
\begin{corollary} \label{cor:GV}
For any fixed $D \ge 1$, we can obtain, via the CWS code construction, pure $[[n^D,k,2D+1]$ QECCs with $k = n^D - \lceil 2D^2(2D+1) \log_2 n\rceil$ for all $n \ge 8$. 
\end{corollary}
\begin{proof}
Let $N = n^D$, $n \ge 8$. By a straightforward application of the Gilbert-Varshamov bound for binary linear codes, there exists a binary linear code, $\mathcal{C}$, of blocklength $N$, dimension $k = N - \lceil 2D(2D+1) \log_2 N\rceil$ and minimum distance $d > 2D(2D+1)$. The claimed pure QECC is obtained by applying the construction in Corollary~\ref{cor:cluster_CWS} to this code $\mathcal{C}$ and the cluster state $\ket{\mathrm{G}}$ associated with the $D$-dimensional rectangular lattice $\Lambda_{n,n,\ldots,n}$.
\end{proof}

The lower bound of $m(m+1)$ on the minimum distance of the linear code $\mathcal C$ in Theorem~\ref{theorem: distanceBased} is a sufficient, but not necessary, condition for the conclusion of the theorem to hold. Indeed, Theorem~\ref{theorem_1} shows that only a weaker lower bound of $m+1$ is necessary (note that the minimum distance of $\mathcal C$ bounds from above the minimum symplectic weight of a non-zero vector in $\mathcal W$). It would thus be useful to try to find linear codes $\mathcal C$ with minimum distance that is of the order of $m$ rather than of the order of $m^2$, from which we can still obtain the same conclusion as in Theorem~\ref{theorem: distanceBased}. Linear codes with a less stringent requirement on minimum distance will have a larger dimension $k$, which directly translates to a larger dimension for the resulting QECCs. 

Given an $m$-uniform graph state $\ket{\mathrm{G}}$, to weaken the minimum distance requirement on the tent peg code $\mathcal C$, we need to incorporate the condition on $\mathcal W$ in Theorem~\ref{theorem_1} into the design of the tent peg code. While this seems difficult to do for arbitrary $m$-uniform graph states, we show how this strategy can be successfully implemented for the cluster states obtained from $1$- and $2$-dimensional lattices. 

\subsection{Pure $[[n,k,3]]$ QECCs from $1$-dimensional cluster states}

Consider the cluster state on $n$ qubits associated with the $1$-dimensional lattice $\Lambda_n$. It can be verified that such a cluster state is $2$-uniform for $n\geq 5$. Note that, in this case, Corollary~\ref{cor:cluster_CWS} requires a minimum distance of $d > 6$ for the tent peg code. Instead, via a direct application of Theorem~\ref{theorem_1}, we will construct a pure $[[n,k,3]]$ QECC by requiring the $[n,k]$ tent peg code $\mathcal{C}$ to be such that $\mathcal{W}$ does not contain a nonzero vector of symplectic weight up to 2. Equivalently, the code $\mathcal{C}$ should be such that we do not get a vector in $\mathcal{V}_\mathcal{S}$ by adding a binary vector of symplectic weight at most 2 to a vector $\mathbf{c}^\prime \in \mathcal{C}^\prime$. This demands the code $\mathcal{C}$ to have a minimum distance of at least 3 and to not contain codewords of the following form:
  \begin{gather}
  \text{\rule[0.25em]{1cm}{0.5pt}} \, 1 \ast \ast \, 1 \,  \text{\rule[0.25em]{1cm}{0.5pt}}   \label{pattern1} \\
  \text{\rule[0.25em]{1cm}{0.5pt}} \, 1 \ast 0 \ast 1 \, \text{\rule[0.25em]{1cm}{0.5pt}}   \label{pattern2} \\
  \text{\rule[0.25em]{1cm}{0.5pt}} \, 1 \ast 1  \, \text{\rule[0.25em]{1cm}{0.5pt}}  \, 1 \ast 1 \, \text{\rule[0.25em]{1cm}{0.5pt}}   \label{pattern3} \\
  \text{\rule[0.25em]{1cm}{0.5pt}} \, 1 \ast 1  \, \text{\rule[0.25em]{1cm}{0.5pt}}  \, \ast \, \text{\rule[0.25em]{1cm}{0.5pt}} \label{pattern4}
 \end{gather}
Here, each `\rule[0.25em]{1cm}{0.5pt}' represents a run of 0s, possibly empty, and each $\ast$ can be a $0$ or a $1$. Also, cyclic shifts of such codewords are not allowed, i.e., the coordinates of these codewords are assumed to lie on a circle. We give a construction of a tent peg code satisfying the above conditions.

Let $r \ge 2$ be an integer. Set $n = 2^{2r}-1$, and $b = (2^{2r}-1)/3$. 
Consider the binary code $\mathcal{C}_r$ with parity-check matrix
\begin{equation}
H = \begin{bmatrix} 
1 & \alpha & \alpha^2 & \alpha^3 & \cdots & \alpha^{n-1} \\
1 & \alpha^b & {(\alpha^b)}^2 & {(\alpha^b)}^3 & \cdots & {(\alpha^b)}^{n-1} 
\end{bmatrix}
\label{def:H}
\end{equation}
where $\alpha$ is a primitive element in $\mathbb{F}_{2^{2r}}$. By choice of $b$, the second row of the $H$ matrix above is simply
$$
[1 \ \ \ \alpha^b \ \ \ \alpha^{2b} \ \ \ 1 \ \ \ \alpha^b \ \ \ \alpha^{2b} \ \ \ \cdots \ \ \ 1 \ \ \ \alpha^b \ \ \ \alpha^{2b}]
$$
We also have $1 + \alpha^b + \alpha^{2b} = \frac{\alpha^{3b} - 1} {\alpha^b - 1} = 0$, so $\alpha^b$ has $x^2+x+1$ as its minimal polynomial over $\mathbb{F}_2$.

The code $\mathcal{C}_r$ is the cyclic code generated by $g(x) = M_{\alpha}(x) M_{\alpha^b}(x)$, where $M_{\alpha}(x)$ and $M_{\alpha^b}(x)$ are the minimal polynomials over $\mathbb{F}_2$ of $\alpha$ and $\alpha^b$, respectively. Thus, $\deg g(x) = 2r + 2$, so that 
\begin{equation*}
\begin{split}
\dim(\mathcal{C}_r) &= n-\deg g(x) \\
&= 2^{2r}-2r-3.
\end{split}
\end{equation*}
Moreover, $d_{\min}(\mathcal{C}_r) \ge 3$, as the row of the $H$ matrix consists of distinct and nonzero elements of $\mathbb{F}_{2^{2r}}$. 

\begin{remark}
Note that $b = (4^r-1)/3 = \sum_{i=0}^{r-1} 4^i \equiv r \pmod{3}$, since $4^i \equiv 1 \pmod{3}$. So, if we assume that $r$ is not a multiple of $3$, we get $b \equiv 1 \text{ or } 2 \pmod{3}$. In this case, we obtain that 
$$
\mathbf{c} = (1,\underbrace{0,\ldots,0}_{b-1 \, 0\text{s}},1,\underbrace{0,\ldots,0}_{b-1 \, 0\text{s}},1,\underbrace{0,\ldots,0}_{b-1 \, 0\text{s}})
$$
is a codeword of $\mathcal{C}_r$, since we have 
\begin{equation*}
\begin{split}
    H \mathbf{c}^T &= [1+\alpha^b + \alpha^{2b}, \ \ 1 + ({\alpha^b})^b + ({\alpha^b})^{2b}]^T\\ &= [0 \ \ 0]^T.
\end{split}
\end{equation*}
Indeed, the fact that $b \equiv 1 \text{ or } 2 \pmod{3}$ implies that $\{({\alpha^b})^b, ({\alpha^b})^{2b}\} = \{\alpha^b,\alpha^{2b}\}$, so that $1 + ({\alpha^b})^b + ({\alpha^b})^{2b} = 0$. Thus, $\mathcal{C}_r$ contains a codeword of weight equal to $3$, which proves that $d_{\min}(\mathcal{C}_r) = 3$ whenever $r$ is not a multiple of $3$.

\end{remark}

\medskip

Our construction of pure QECCs from $1$-dimensional cluster states is summarized in the theorem stated next. In the statement of the theorem, $\Log$ stands for the discrete logarithm in the field $\mathbb{F}_{2^{2r}}$ with respect to a fixed primitive element $\alpha$, defined as $\Log(\beta) = \ell \ \Longleftrightarrow \ \beta = \alpha^{\ell}$.

\begin{theorem}\label{theorem3}
Let $\ket{\mathrm{G}}$ be the cluster state associated with the $1$-dimensional rectangular lattice $\Lambda_n$, {where $n = 2^{2r}-1$ for some $r \ge 2$}. {To construct the $H$ matrix in \eqref{def:H}, choose a primitive element $\alpha \in \mathbb{F}_{2^{2r}}$ with the property that $\Log(1+\alpha) \not\equiv 2 \pmod{3}$.} Then, for the binary cyclic code $\mathcal{C}_r$ {with parity-check matrix $H$}, the quantum code  $\mathcal{Q}_r := \text{span}\bigl(\{ Z_{\mathbf{c}}\ket{\mathrm{G}}: \mathbf{c} \in \mathcal{C}_r \}\bigr)$ is a pure $[[2^{2r}-1,2^{2r}-2r-3,3]]$ QECC.
\end{theorem}

A proof of this theorem is given in Appendix \ref{appendix:theorem3}. For the construction in the theorem to work, we need a primitive element $\alpha$ in $\mathbb{F}_{2^{2r}}$ satisfying $\Log(1+\alpha)\not\equiv 2 \pmod 3$. {Such an element always exists for $r\geq 2$, as stated in the proposition below. A proof is given in Appendix~\ref{conjecture_proof}.}

\begin{proposition} \label{resolved_conjecture}
For any integer $r\geq 2$, there exists a primitive element $\alpha \in \mathbb{F}_{2^{2r}}$ such that $\Log(1+\alpha) \not\equiv 2 \pmod 3$.
\end{proposition}

At blocklengths $n = 2^{2r}-1$, $r \ge 2$, the parameters of the code family obtained in Theorem~\ref{theorem3} are better than those of the $[[2^{2r}-1, 2^{2r}- 4r - 1, 3]]$ quantum Hamming code \cite{eczoo_quantum_hamming_css}, and are the same as those of a putative code obtained from the $[[2^{2r},2^{2r}-2r-2,3]]$ Gottesman code \cite{Gottesman_code},\cite{eczoo_quantum_gottesman} by shortening at one qubit. 

For blocklengths $15$, $63$ and $255$, the dimensions of codes obtained in Theorem~\ref{theorem3} match the dimensions of the best known codes of minimum distance 3 in the table of QECCs \cite{Grassl:codetables}, which only record codes of blocklength up to $256$.

Since the QECCs constructed in Theorem~\ref{theorem3} are pure codes, they need to respect the quantum Hamming bound \cite[Section~10.3.4]{Nielsen_Chuang_2010}:
\begin{equation*}
 Q := \frac{\sum_{i=0}^{\lfloor (d-1)/2 \rfloor} \binom{n}{i} 3^i}{2^{n-k}} \leq 1.
\end{equation*}
For the codes $\mathcal{Q}_r$ in Theorem~\ref{theorem3}, as $r\rightarrow \infty$, we have $Q \rightarrow 3/4$.


\subsection{Pure $[[n^2,k,5]]$ QECCs from $2$-dimensional cluster states}

When we start with a $2$-dimensional cluster state $\Lambda_{n,n}$ with $n \ge 8$, Corollary~\ref{cor:GV} (applied with $D=2$) shows that we can obtain pure $[[n^2,k,5]]$ QECCs with $k = n^2 - \lceil{40 \log_2 n}\rceil$. In this section, we show that for integers $n$ of the form $2^{4r}-1$, $r \ge 1$, we can do significantly better using $2$-dimensional cyclic codes. Note that for $n = 2^{4r}-1$, the $[[n^2,k,5]]$ QECCs from Corollary~\ref{cor:GV} have dimension $k = n^2 - 160r + O(1)$. The $[[n^2,k,5]]$ QECCs we construct via Theorem~\ref{thm:2D} in this section have dimension $k = n^2 - 32r - O(1)$.

Let $\mathbb{F}_{2}^{n \times n}$ denote the $n^2$-dimensional vector space over $\mathbb{F}_2$ whose elements are $n \times n$ arrays of the form $[a_{i,j}]_{0 \leq i,j \leq n-1}$. We identify such arrays with bivariate polynomials $\sum_{i=0}^{n-1} \sum_{j=0}^{n-1} a_{i,j} x^i y^j$, resulting in the vector space isomorphism $\mathbb{F}_2^{n \times n} \longleftrightarrow \mathbb{F}_2[x,y] \slash \langle x^n-1,~y^n-1 \rangle$.

A $2$-dimensional cyclic code is a subspace $C$ of $\mathbb{F}_{2}^{n \times n}$ with the property that if an array $[a_{i,j}]_{0 \leq i,j \leq n-1}$ is in $C$, then for any pair of integers $s,t$, the array $[a_{i+s,j+t}]_{0 \leq i,j \leq n-1}$ is also in $C$, where $i+s$ and $j+t$ are taken modulo $n$. By a routine argument, a subspace $C \subset \mathbb{F}_2^{n \times n}$ is $2$-dimensional cyclic iff $C$ is an ideal in the polynomial ring $\mathbb{F}_2[x,y] \slash \langle x^n-1,~y^n-1 \rangle$  \cite{guneri_thesis}.

Let $\alpha$ be a primitive element of $\mathbb{F}_{2^s}$ (for some $s \ge 1$), and let $n = 2^s-1$. Consider the set
 $$\Omega=\{(\alpha^i,\alpha^j)~:~0 \leq i,j \leq n-1\}.$$
 Note that if a polynomial $f(x,y) \in \mathbb{F}_2[x,y]$ vanishes at $(\alpha^i,\alpha^j)$ then it also vanishes at $(\alpha^{2i},\alpha^{2j}),~(\alpha^{2^2 i}, \alpha^{2^2 j}),$ $ \dots, (\alpha^{2^{m-1} i},\alpha^{2^{m-1} j})$, where $m$ is the least common multiple of $[\mathbb{F}_2(\alpha^i):\mathbb{F}_2]$ and $[\mathbb{F}_2(\alpha^j):\mathbb{F}_2]$. Let $[(\alpha^i,\alpha^j)]:=\{ (\alpha^i,\alpha^j),(\alpha^{2i},\alpha^{2j}),(\alpha^{2^2 i}, \alpha^{2^2 j}),\dots,(\alpha^{2^{m-1} i},\alpha^{2^{m-1} j})\}$ be the $\mathbb{F}_2$-conjugacy class of $(\alpha^i,\alpha^j)$ in $\Omega$. For $U \subseteq \Omega$, let $[U]$ denote the union of $\mathbb{F}_2$ conjugacy classes of elements in $U$.
 
 \begin{definition}[Definition~2.4 in \cite{guneri_thesis}]
     For $U \subseteq \Omega$, the ideal corresponding to $U$ is defined as 
     $$I(U) := \{f(x,y) \in \mathbb{F}_2[x,y]~:~f(a)=0~ \forall \, a \in U\}.$$
      \end{definition}

Since $n = 2^s - 1$, we have that $x^n-1, y^n-1 \in I(U)$ for any $U \subset \Omega$, and therefore, $C_U := I(U)\slash \langle x^n-1,~y^n-1 \rangle$ is a $2$-dimensional cyclic code. The following result, due to G\"uneri \cite{guneri_thesis}, gives the dimension, over $\mathbb{F}_2$, of the code $C_U$.
     
\begin{theorem}[Theorem~2.17 in \cite{guneri_thesis}] \label{thm:guneri}
    For $U \subseteq \Omega$, the dimension, over $\mathbb{F}_2$, of the $2$-dimensional cyclic code $C_U:=I(U)\slash \langle x^n-1,~y^n-1 \rangle$ is given by
    $$\dim_{\mathbb{F}_2}(C_U)= n^2 - \bigl|[U]\bigr|.$$
\end{theorem}

We are now in a position to describe our construction of pure $[[n^2,k,5]]$ QECCs. We make use of $2$-dimensional cyclic codes constructed over the fields $\mathbb{F}_{2^{4r}}$, $r \ge 1$. Note that $15|(2^{4r}-1)$, so that third and fifth roots of unity exist in $\mathbb{F}_{2^{4r}}$. Let $\gamma$ and $\beta$ respectively denote a third and fifth root of unity in $\mathbb{F}_{2^{4r}}$. Let us choose $U \subset \Omega$ to be 

\begin{equation} \label{def:U}
\begin{split} 
    U &=\{(1,1),~(\alpha,\alpha),~(\alpha,\alpha^{-1}),~(\alpha,1),~(\alpha^3,1),~(1,\alpha),\\
    &{\ \ \ \ }~(1,\alpha^3),~(1,\gamma),~(\beta,\beta^2),~(\alpha,\alpha^2),~(\alpha^3,\alpha^{-3})\}.
\end{split}
\end{equation}

Take $n = 2^{4r}-1$, and let $C_U \subset \mathbb{F}_2^{n \times n}$ be the $2$-dimensional cyclic code defined by $U$. Note that $f(x,y)=(x+y)(1+x)(1+y)(1+xy)(x^2+y)=x^3+x^4+x^5 y+x^5 y^2+xy+x^3 y^2+x^2 y^2+x^4 y^3+y^2+y^3+x y^4+x^2 y^4 \in I(U)$. Hence, the minimum distance of $C_U$ does not exceed $12$, which means that the code does not satisfy the condition $d_{\min} > 2D(2D+1)$, so that Corollaries~\ref{cor:cluster_CWS} and \ref{cor:GV} do not apply to this code.

\begin{claim} \label{claim:dimCU}
    We have $\dim_{\mathbb{F}_2}(C_U) = {(2^{4r}-1)}^2 - 32r - 7$.
\end{claim}
\begin{proof}
It is straightforward to verify that the $\mathbb{F}_2$-conjugacy classes of distinct elements of $U$ in \eqref{def:U} are disjoint. We moreover have $|[(1,1)]|=1$, $|[(\alpha,\alpha)]|=4r$, $|[(\alpha,\alpha^{-1})]|=4r$, $|[(\alpha,1)]|=4r$, $|[(\alpha^3,1)]|=4r$, $|[(1,\alpha)]|=4r$, $|[(1,\alpha^3)]|=4r$, $|[(1,\gamma)]|=2$, $|[(\beta,\beta^2)]|=4$, $|[(\alpha,\alpha^2)]|=4r$, and $|[(\alpha^3,\alpha^{-3})]|=4r$. Thus, $\bigl|[U]\bigr| = 8 \times (4r) + 7$, and the claim now follows from Theorem~\ref{thm:guneri}.
\end{proof}

\begin{example}
    For $r=1$, let $\alpha$ be a primitive element of $\mathbb{F}_{16}$. Then $\gamma=\alpha^5$ and $\beta=\alpha^3$. Thus,
    \begin{equation*}
    \begin{split}
    U &=\{(1,1),~(\alpha,\alpha),~(\alpha,\alpha^{-1}),~(\alpha,1),~(\alpha^3,1),~(1,\alpha),\\
    &{\ \ \ }~(1,\alpha^3),~(1,\alpha^5),~(\alpha^3,\alpha^6),~(\alpha,\alpha^2),~(\alpha^3,\alpha^{-3})\}.
    \end{split}
    \end{equation*}
 Here $C_U$ is a binary linear code given by the null space of a $11 \times 225$ matrix $H$ with columns
 \begin{equation*}
 \begin{split}
 \big[ &1~~\alpha^{l_1+l_2}~~\alpha^{l_1-l_2}~~\alpha^{l_1}~~\alpha^{3l_1}~~\alpha^{l_2}~~\alpha^{3l_2}~~\alpha^{5l_2}~~\alpha^{3l_1+6l_2}\\ &~~\alpha^{l_1+2l_2}~~\alpha^{3(l_1-l_2)}\big]^{T},
 \end{split}
 \end{equation*}
 where $0\leq l_1,l_2 \leq 14$. Then $k=dim_{\mathbb{F}_2}(C_U)=186$ and minimum distance of $C_U$ is $d_{min}(C_U)=6$ (verified by Magma \cite{Magma}). An example of a minimum weight codeword of $C_U$ is $f:=1+x+x^4+x^5+x^6+x^9$.
\end{example}

Our construction of pure QECCs from $2$-dimensional cluster states and $2$-dimensional cyclic codes is summarized by the following theorem. 

\begin{theorem}\label{thm:2D}
Let $\ket{\mathrm{G}}$ be the cluster state associated with the $2$-dimensional rectangular lattice $\Lambda_{n,n}$, where $n = 2^{4r}-1$ for some $r \ge 1$. Then, for the $2$-dimensional cyclic code $C_U \subset \mathbb{F}_2^{n \times n}$ defined from the set $U$ in \eqref{def:U}, the quantum code  $\mathcal{Q}_U := \text{span}\bigl(\{ Z_{\mathbf{c}}\ket{\mathrm{G}}: \mathbf{c} \in C_U\}\bigr)$ is a pure $[[(2^{4r}-1)^2, (2^{4r}-1)^2 - 32r-7, 5]]$ QECC.
\end{theorem} 

Since the code $C_U$ does not satisfy the condition $d_{\min} > 2D(2D+1)$, the proof of the theorem relies on a direct verification of the necessary and sufficient condition in Theorem~\ref{theorem3}. We present the proof in Appendix~\ref{appendix:thm2D}. 

We are not aware of any infinite families of pure QECCs of distance $5$ against which the parameters of our code family from Theorem~\ref{thm:2D} can be compared. A direct comparison of the code constructed in Theorem~\ref{thm:2D} with the best known QECC from the code table \cite{Grassl:codetables} is only possible for $r=1$. In this case, our construction yields a pure $[[225,186,5]]$ QECC, whereas the code table \cite{Grassl:codetables} lists a $[[225,186,7]]$ QECC. The latter QECC has better minimum distance but the table does not record if the code is pure.

\section{Measurement-Based Encoding and Recovery Protocols}
\label{sec: encoding}

In this section, we explain our tent peg protocols for encoding and recovery of logical information for additive CWS codes. Sections~\ref{sec:enc} and \ref{sec:seq_enc} explain the encoding protocol, which uses $n$ physical qubits that are initially in a graph state $\ket{\mathrm{G}}$. Sections~\ref{sec:dec} and \ref{sec:partial_dec} explain the protocol for recovering the logical qubits from the encoded state. 

\subsection{Encoding} \label{sec:enc}

\begin{figure}[h]
    \centering

\tikzset{every picture/.style={line width=0.75pt}} 

\begin{tikzpicture}[x=0.75pt,y=0.75pt,yscale=-0.8,xscale=0.8]

\draw  [fill={rgb, 255:red, 0; green, 0; blue, 0 }  ,fill opacity=1 ] (19.5,230.75) .. controls (19.5,227.85) and (21.85,225.5) .. (24.75,225.5) .. controls (27.65,225.5) and (30,227.85) .. (30,230.75) .. controls (30,233.65) and (27.65,236) .. (24.75,236) .. controls (21.85,236) and (19.5,233.65) .. (19.5,230.75) -- cycle ;
\draw  [fill={rgb, 255:red, 0; green, 0; blue, 0 }  ,fill opacity=1 ] (70.5,230.75) .. controls (70.5,227.85) and (72.85,225.5) .. (75.75,225.5) .. controls (78.65,225.5) and (81,227.85) .. (81,230.75) .. controls (81,233.65) and (78.65,236) .. (75.75,236) .. controls (72.85,236) and (70.5,233.65) .. (70.5,230.75) -- cycle ;
\draw  [fill={rgb, 255:red, 0; green, 0; blue, 0 }  ,fill opacity=1 ] (119.5,230.75) .. controls (119.5,227.85) and (121.85,225.5) .. (124.75,225.5) .. controls (127.65,225.5) and (130,227.85) .. (130,230.75) .. controls (130,233.65) and (127.65,236) .. (124.75,236) .. controls (121.85,236) and (119.5,233.65) .. (119.5,230.75) -- cycle ;
\draw  [fill={rgb, 255:red, 0; green, 0; blue, 0 }  ,fill opacity=1 ] (180.5,230.75) .. controls (180.5,227.85) and (182.85,225.5) .. (185.75,225.5) .. controls (188.65,225.5) and (191,227.85) .. (191,230.75) .. controls (191,233.65) and (188.65,236) .. (185.75,236) .. controls (182.85,236) and (180.5,233.65) .. (180.5,230.75) -- cycle ;
\draw  [fill={rgb, 255:red, 0; green, 0; blue, 0 }  ,fill opacity=1 ] (229.5,230.75) .. controls (229.5,227.85) and (231.85,225.5) .. (234.75,225.5) .. controls (237.65,225.5) and (240,227.85) .. (240,230.75) .. controls (240,233.65) and (237.65,236) .. (234.75,236) .. controls (231.85,236) and (229.5,233.65) .. (229.5,230.75) -- cycle ;
\draw [line width=1.5]    (24.75,230.75) -- (124.75,230.75) ;
\draw [line width=1.5]  [dash pattern={on 3.75pt off 3.75pt}]  (124.75,230.75) -- (234.75,230.75) ;
\draw  [fill={rgb, 255:red, 0; green, 0; blue, 0 }  ,fill opacity=1 ] (140.5,114.75) .. controls (140.5,111.85) and (142.85,109.5) .. (145.75,109.5) .. controls (148.65,109.5) and (151,111.85) .. (151,114.75) .. controls (151,117.65) and (148.65,120) .. (145.75,120) .. controls (142.85,120) and (140.5,117.65) .. (140.5,114.75) -- cycle ;
\draw [line width=0.75]    (145.75,114.75) -- (75.75,230.75) ;
\draw [line width=0.75]    (145.75,114.75) -- (185.75,230.75) ;
\draw [line width=0.75]    (145.75,114.75) -- (124.75,230.75) ;

\draw (24.75,223.5) node [anchor=south] [inner sep=0.75pt]   [align=left] {$\displaystyle 1$};
\draw (234.75,223.5) node [anchor=south] [inner sep=0.75pt]   [align=left] {$\displaystyle n$};
\draw (120.75,223.5) node [anchor=south] [inner sep=0.75pt]   [align=left] {$\displaystyle 3$};
\draw (74.75,223.5) node [anchor=south] [inner sep=0.75pt]   [align=left] {$\displaystyle 2$};
\draw (130,253.27) node [anchor=south] [inner sep=0.75pt]   [align=left] {$\displaystyle \underbrace{\ \ \ \ \ \ \ \ \ \ \ \ \ \ \ \ \ \ \ \ \ }$};
\draw (127.39,252) node [anchor=north] [inner sep=0.75pt]   [align=left] {$\displaystyle A$};
\draw (99,68) node [anchor=north west][inner sep=0.75pt]   [align=left] {$\displaystyle \alpha \ket{0} +\beta \ket{1}$};
\draw (99,88) node [anchor=north west][inner sep=0.75pt]   [align=left] {Input qubit};

\end{tikzpicture}
\caption{{Encoding one logical qubit: the input qubit contains the information to be encoded, and the bottom array of $n$ qubits is initially in a graph state. Each solid line connecting two qubits represents a controlled-$Z$ gate.}}

    \label{one ancilla}
\end{figure}

For clarity, we first revisit the protocol for encoding a single qubit \cite{sowrabh} into a $2$-dimensional CWS code $\text{span}\{\ket{\mathrm{G}}, Z_{A}\ket{\mathrm{G}}\}$. As shown in Figure~\ref{one ancilla}, we have $n$ physical qubits, referred to hereafter as target qubits, prepared in the graph state $\ket{\mathrm{G}}$, and an {input qubit} in the state $\alpha \ket{0} + \beta \ket{1}$, which needs to be encoded. We use the {input qubit} as the control qubit of a set of controlled-$Z$ ($CZ$) gates that act on a subset $A$ of the target qubits. These $CZ$ gates bring the joint state of the system into the form given below:
\begin{align*}
    \prod_{i\in A} CZ_{i} \big( \alpha\ket{0} +\beta \ket{1} \big)\ket{\mathrm{G}}  &= {\alpha\ket{0}\ket{G} + \beta\ket{1} Z_A\ket{G}} \\
    &= \frac{1}{\sqrt{2}} \Big(\ket{+}\ket{\mathrm{e}_{+}} + \ket{-}\ket{\mathrm{e}_{-}}\Big),
\end{align*}
where $CZ_{i}$ is a $CZ$ gate acting on the $i^{th}$ target qubit, and $\ket{\mathrm{e}_{\pm}} = \alpha\ket{\mathrm{G}} \pm \beta Z_A \ket{\mathrm{G}}$. {The second line is obtained by expressing $\ket{0}$ and $\ket{1}$ in the $X$-basis $\{\ket{+}, \ket{-}\}$.} Applying a projective measurement on the {input qubit} in the $X$-basis transforms the target qubits to either the desired encoding $\ket{\mathrm{e}} = \ket{\mathrm{e}_+}$ or a different state $\ket{\mathrm{e}_{-}}$ depending on the measurement outcome. It is easy to see that $\ket{\mathrm{e}_{-}}$ can be transformed to $\ket{\mathrm{e}}$ by performing a logical-$Z$ operation. Note that the graph stabilizer generator, $S_i$, corresponding to any qubit $i \in A$, anti-commutes with $Z_A$. Therefore, $S_{i}\ket{\mathrm{e}_{-}} = \ket{\mathrm{e}}$, and hence, $S_i$ is a logical-$Z$ operator for the code. Since the stabilizers $S_i \in \mathcal{P}_n$ consist of tensor products over local unitaries, they are not difficult to implement.

This measurement-based approach can be generalized to construct higher-dimensional codes following the visual guide in Figure~\ref{kAncillas}. We again start with $n$ physical qubits, indexed by $[n]$, that are jointly in a graph state $\ket{\mathrm{G}}$. But we now have $k$ {input qubits} (shown at the top of the figure), with the $i$-th {input qubit} controlling $CZ$ gates that act on a subset $A_i \subseteq [n]$ of the $n$ physical qubits, which we will again refer to as target qubits. The figure motivates the ``tent peg'' nomenclature we use for the encoding protocol, as the figure brings to mind a row of tents whose tops are the {input qubits} and whose pegs are the target qubits with which the {input qubits} interact. It should be clarified, however, that the ``tent peg sets'' $A_i$ need not be disjoint.

\begin{figure}[h]
    \centering
    \includegraphics[width = 0.5 \textwidth]{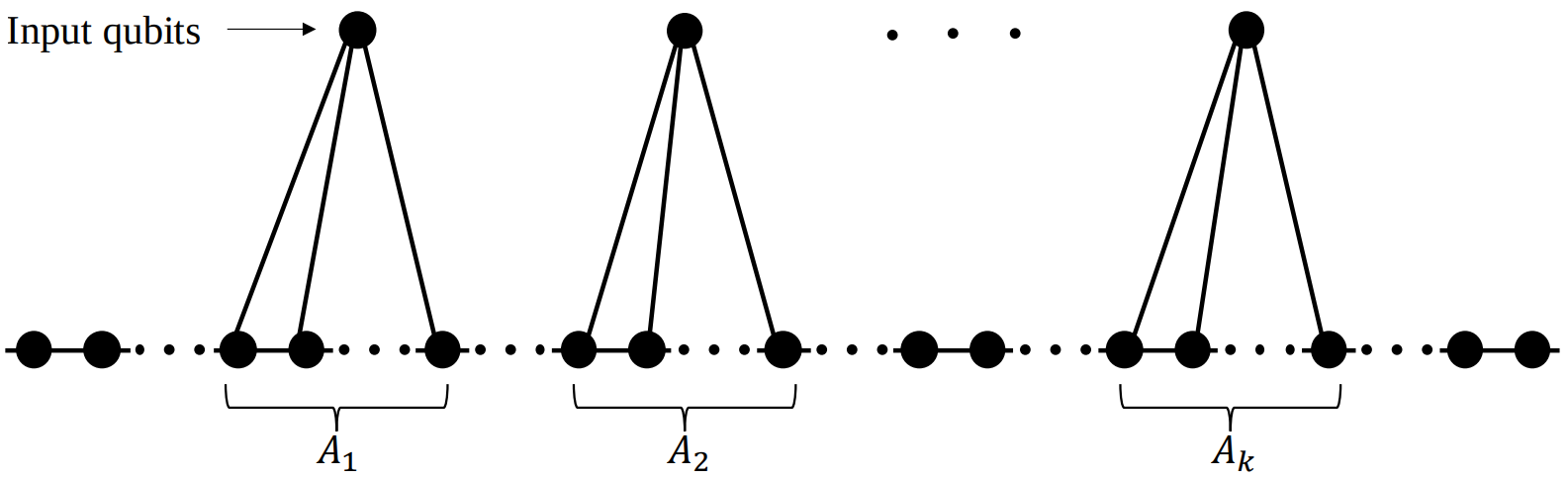}
    \caption{Tent peg protocol: The top horizontal array of dots represents $k$ {input qubits} holding logical information. The bottom horizontal array of dots represents the $n$ target qubits, initially in a graph state. Each solid line between a pair of qubits represents a controlled-$Z$ gate between the two qubits.}
    \label{kAncillas}
\end{figure}

The choice of the tent peg sets $A_i$ is governed by an $[n,k]$ binary linear code $\mathcal{C}$, which we refer to as the ``tent peg code''. Let $\mathbf{A}$ be a $k \times n$ generator matrix for $\mathcal C$, which, without loss of generality, we will assume to be in systematic form. Thus, $\mathbf{A} = [\mathbf{I} \mid \mathbf{B}]$, where $\mathbf{I}$ is the $k \times k$ identity matrix and $\mathbf{B}$ is a binary $k \times (n-k)$ matrix. Denoting the rows of $\mathbf{A}$ by $\mathbf{a}_1, \ldots, \mathbf{a}_k$, we take $A_i = \text{supp}(\mathbf{a}_i)$ for $i=1,\ldots,k$. 

Let the $k$ logical qubits to be encoded be in an arbitrary pure state 
\begin{equation} \label{eq:logical_qubits}
\ket{\phi} = \sum_{\mathbf{x}\in\{0,1\}^k} a_{\mathbf{x}}\ket{\mathbf{x}}, 
\end{equation}
We prepare $k$ {input qubits} in the state $\sum_{\mathbf{x}\in\{0,1\}^k} a_{\mathbf{x}}\ket{\mathbf{x}}_{\mathrm a}$, where the subscript $\mathrm a$ is used for clarity to identify the {input qubits}. Then, the {input qubits} and target qubits are jointly in the state
\begin{align}\label{Before encoding}
    \ket{\psi} = \sum_{\mathbf{x} \in \{0,1\}^k}{a_{\mathbf{x}}\ket{\mathbf{x}}_{\mathrm a} \ket{\mathrm{G}}}.
\end{align}
Applying the $CZ$ gates controlled by the {input qubits}, the joint state of the $k$ {input qubits} and $n$ target qubits becomes
\begin{align*}
    \ket{\psi^\prime} = \sum_{\mathbf{x}\in \{0,1\}^k}{a_{\mathbf{x}}\ket{\mathbf{x}}_{\mathrm a} \Big(Z_{\mathbf{x}\mathbf{A}}\ket{\mathrm{G}}\Big)}.
\end{align*}
Here, the subscript $\mathbf{x}\mathbf{A}$ on $Z$ refers to the binary vector obtained as the modulo-$2$ product of the vector $\mathbf{x}$ and the matrix $\mathbf{A}$.

Next, the {input qubits} are measured sequentially in the $X$-basis. The joint state of the $n$ target qubits is then transformed to
\begin{align*} 
    \ket{\mathrm{e}^{\prime}} = \sum_{\mathbf{x} \in \{0,1\}^k}{{(-1)}^{\mathbf{x} \cdot \mathbf{y}} a_{\mathbf{x}}Z_{\mathbf{x}\mathbf{A}}\ket{\mathrm{G}}},
\end{align*}
where $\mathbf{y} \in \{0,1\}^k$ is a vector whose $i^{th}$ entry is $0$ if the $i^{th}$ {input qubit} measurement resulted in the $\ket{+}$ state, and $1$ if the measurement resulted in the $\ket{-}$ state. 

Similar to the case where a single qubit was encoded, for each $i \in \text{supp}(\mathbf{y})$, applying the logical-$Z$ operator, $Z^{L}_i$, corresponding to the $i^{th}$ logical qubit transforms the post-measurement state $\ket{\mathrm{e}^{\prime}}$ to the desired form
\begin{equation} \label{eq:encoding}
\ket{\mathrm{e}}=\sum_{\mathbf{x} \in \{0,1\}^k}{ a_{\mathbf{x}}Z_{\mathbf{x}\mathbf{A}}\ket{\mathrm{G}}}.
\end{equation}
Explicitly, $Z^{L}_i = S_i$, the stabilizer generator associated with the $i$th vertex of the graph underlying $\ket{\mathrm{G}}$. This is because, for each $i \in [k]$, $S_i = X_i Z_{N(i)}$ anti-commutes with $Z_{A_i}$ and commutes with $Z_{A_j}$ for all $j \neq i$. This is due to the fact that we chose the generator matrix $\mathbf{A}$ to be of the form $[\mathbf{I} \mid \mathbf{B}]$, so that only $Z_{A_i}$ contains the operator $Z_i$, for each $i \in [k]$.

We summarize in Algorithm~\ref{alg:enc} below the tent peg protocol for encoding $k$ logical qubits into a CWS code 
\begin{equation}\label{eq:Q_CWS}
\mathcal{Q} := \text{span}\bigl(\{Z_{\mathbf{x}\mathbf{A}} \ket{\mathrm{G}}: \mathbf{x} \in \{0,1\}^k\}\bigr)
\end{equation}
obtained from a graph state $\ket{\mathrm{G}}$ and a tent peg code with a systematic generator matrix $\mathbf{A}$. 
\begin{algorithm}
\caption{The tent peg protocol for encoding $k$ logical qubits in a CWS code.}\label{alg:enc}
\hspace*{\algorithmicindent} {\textbf{Input:} $k$ input qubits in the logical state to be encoded} \\
\hspace*{\algorithmicindent} {\textbf{Output:} $n$ physical qubits in the code state} 
\begin{enumerate}
    \item Prepare $n$ physical qubits in the graph state $\ket{\mathrm{G}}$
    \item For $i = 1 \text{ to } k$:
    \begin{itemize}
        \item[] apply a $CZ$ gate with the $i$-th {input qubit} as control and the subset $A_i$ of the $n$ physical qubits as the target
    \end{itemize}
    \item For $i = 1 \text{ to } k$:
    \begin{enumerate}
        \item measure the $i$-{th} {input qubit} in the $X$-basis
        \item if the measurement outcome is `$-1$', apply $Z_i^L = S_i$ on the physical qubits
    \end{enumerate}
\end{enumerate}
\end{algorithm}

Alternatively, instead of transforming the post-measurement state to the desired form, we may use $k$ classical bits to store the measurement outcome, along with the state $\ket{\mathrm{e}^\prime}$. A party without access to the classical bits cannot decode the correct logical information only from $\ket{\mathrm{e}^\prime}$.

\subsection{Sequential encoding} \label{sec:seq_enc}
Among the $k$ logical qubits, suppose that there is a qubit that is not entangled with the remaining $k-1$ logical qubits. Without loss of generality, assume that the $k$th logical qubit is separable from (i.e., not entangled with) the first $k-1$ logical qubits. Then, the entire set of $k$ logical qubits can be sequentially encoded by first applying the tent peg protocol to the initial set of $k-1$ logical qubits, and then encoding the $k$th logical qubit. To describe this in more detail, let us assume that the logical qubits are in a product state 
\begin{equation} \label{eq:productstate} 
\ket{\phi} = \left(\sum_{\mathbf{x}' \in \{0,1\}^{k-1}} a_{\mathbf{x}'} \ket{\mathbf{x}'} \right) \otimes (b_0\ket{0} + b_1\ket{1}),
\end{equation}
where $b_0 \ket{0} + b_1\ket{1}$ is the state of the $k$th logical qubit. Comparing with \eqref{eq:logical_qubits}, we see that $a_{\mathbf{x}'0} = a_\mathbf{x}' b_0$ and $a_{\mathbf{x}'1} = a_{\mathbf{x}'}b_1$. When we apply the protocol in Algorithm~\ref{alg:enc} to the first $k-1$ logical qubits, we get an intermediate encoded state
\begin{equation} \label{eq:ketf}
    \ket{\mathrm{f}} = \sum_{\mathbf{x}' \in \{0,1\}^{k-1}}{ a_{\mathbf{x}'}Z_{\mathbf{x}'\mathbf{A}^\prime}\ket{\mathrm{G}}},
\end{equation}
where $\mathbf{A}^\prime$ is the submatrix of $\mathbf{A}$ consisting of all but the last row of $\mathbf{A}$. Observe that the state $\ket{\mathrm{f}}$ belongs to the subcode 
\begin{equation} \label{eq:Q'}
\mathcal{Q}' :=  \text{span}\bigl(\{Z_{\mathbf{x}'\mathbf{A}^\prime}\ket{\mathrm{G}}: \mathbf{x}' \in \{0,1\}^{k-1}\}\bigr)
\end{equation}
of the overall CWS code $\mathcal Q$ in \eqref{eq:Q_CWS}.
With the $n$ target qubits in state $\ket{\mathrm{f}}$, if we further apply the encoding protocol for the $k$th logical qubit (i.e., apply $CZ$s controlled by the $k$th logical qubit, then measure that qubit in the $X$-basis, and finally, if the measurement outcome is $-1$, apply $X_k^L = S_k$ on the target qubits), then it can be verified that we will get the final encoding to be $\ket{\mathrm{e}}$ as in~\eqref{eq:encoding}.

\subsection{Recovering the logical qubits}\label{sec:dec}

We propose a similar measurement-based approach to recover the $k$ logical qubits encoded in a state of the CWS code $\mathcal Q$. Note that any correctable errors affecting the physical qubits will first need to be corrected using the standard error-correcting techniques for a stabilizer code before we apply the recovery protocol.

Consider a CWS encoding $\ket{\mathrm{e}}$ of $k$ logical qubits as given in \eqref{eq:encoding}, obtained using an $n$-qubit graph state $\ket{\mathrm{G}}$ and an $[n,k]$ tent peg code with systematic generator matrix $\mathbf{A}$. Start by preparing $k$ {recovery qubits} in the product state $\ket{+}^{\otimes k}$, so that the joint state of the composite system is
\begin{equation*}
    |\psi\rangle = \frac{1}{2^{k/2}}\sum_{\mathbf{x},\mathbf{y}\in\{0,1\}^k} a_{\mathbf{x}} \ket{\mathbf{y}}_{\textsc{r}} Z_{\mathbf{x}\mathbf{A}}\ket{\mathrm{G}},
\end{equation*}
{the subscript `$\textsc{r}$' being used to distinguish the recovery qubits}. Next, as in the encoding protocol, apply $CZ$ gates controlled by the {recovery qubits}, with targets determined by the supports of the rows of $\mathbf A$. This results in the joint state
\begin{equation} \label{eq:phi'}
    \ket{\psi^{\prime}} = \frac{1}{2^{k/2}}\sum_{\mathbf{x},\mathbf{y}\in \{0,1\}^k} a_{\mathbf{x}}\ket{\mathbf{y}}_{\textsc{r}} Z_{(\mathbf{x}+\mathbf{y})\mathbf{A}} \ket{\mathrm{G}}.
\end{equation}
Observe that, if we now apply the projector $\ket{\mathrm G} \bra{\mathrm G}$ to the state of the $n$ physical qubits, then the joint state of the composite system becomes
\begin{equation} \label{eq:post_proj}
    \sum_{\mathbf{x},\mathbf{y}\in \{0,1\}^k} a_{\mathbf{x}} \bra{\mathrm{G}} Z_{(\mathbf{x}+\mathbf{y})\mathbf{A}} \ket{\mathrm{G}} \ket{\mathbf{y}}_{\textsc{r}} \ket{\mathrm{G}}.
\end{equation}
Arguing as in the proof of Proposition~\ref{prop:CWScode}, since the stabilizer of $\ket{\mathrm G}$ is equal to the centralizer and contains no non-identity operators consisting only of Pauli-$Z$s, $\bra{\mathrm{G}}Z_{(\mathbf{x}+\mathbf{y})\mathbf{A}} \ket{\mathrm{G}}$ is non-zero iff $\mathbf{x} = \mathbf{y}$. Therefore, the right-hand side of \eqref{eq:post_proj} simplifies to
\begin{equation}\label{full decoded}
    \sum_{\mathbf{x}\in \{0,1\}^k} a_{\mathbf{x}} \ket{\mathbf{x}}_{\textsc{r}} \ket{\mathrm{G}}.
\end{equation}
Note that this is the same as the state $\ket{\psi}$ in \eqref{Before encoding}. Thus, we have reconstructed the original graph state and extracted the $k$ logical qubits of information onto the $k$ {recovery qubits}.

One way to effect the projection onto $\ket{\mathrm G}$ required to obtain \eqref{eq:post_proj} is as follows. Starting from the state $\ket{\phi'}$ in \eqref{eq:phi'}, we sequentially measure the $n$ independent stabilizer generators $S_1, S_2, \ldots, S_n$. If, for a particular $S_i$ in the sequence, the measurement outcome is `$+1$', then we retain the post-measurement state and measure $S_{i+1}$ next. On the other hand, if the outcome of the measurement is `$-1$', then we pick any Pauli operator that anti-commutes with that $S_i$, but commutes with all other $S_j$, $j \ne i$, and apply that Pauli operator to the post-measurement state. Note that, since $S_i = X_i Z_{N(i)}$, the Pauli operator $Z_i$ (i.e., Pauli-$Z$ acting on the $i$th physical qubit) anti-commutes with only $S_i$ among all the stabilizer generators. The procedure of measuring $S_i$ followed by applying $Z_i$ conditioned on the $S_i$ measurement outcome is mathematically equivalent to the action of the projection operator $\frac{1}{2}(I+S_i)$. Therefore, repeating this for all $S_i$'s results in the action of
\begin{equation} \notag 
    \frac{1}{2^n} \prod_{i=1}^n (I+S_i) = \ket{\mathrm{G}}\bra{\mathrm{G}}
\end{equation} 
on the physical qubits.

\subsection{Partial recovery of logical qubits}  \label{sec:partial_dec}
As in Section~\ref{sec:seq_enc}, we consider here the situation when the $k$th logical qubit is separable from (i.e., not entangled with) the first $k - 1$ logical qubits, as in \eqref{eq:productstate}. In this case, the encoded state $\ket{\mathrm e}$ in \eqref{eq:encoding} can be expressed as 
\begin{equation*}
    \ket{\mathrm{e}} = \sum_{\mathbf{x}'\in \{0,1\}^{k-1}} \bigg(a_{\mathbf{x}'} b_0 \, Z_{\mathbf{x}'\mathbf{A}^\prime}\ket{\mathrm{G}} + a_{\mathbf{x}'} b_1 \, Z_{\mathbf{a}_k}Z_{\mathbf{x}'\mathbf{A}^\prime}\ket{\mathrm{G}} \bigg).
\end{equation*}

We prepare a {recovery qubit} in the $\ket{+}_{\mathrm a}$ state and apply the same set of $CZ$ gates that were used for encoding the $k$th logical qubit. This results in the joint state
\begin{equation*}
    \ket{\psi'} = \frac{1}{\sqrt{2}}\left(\ket{0}_{\mathrm{a}} \ket{\mathrm{e}} + \ket{1}_{\mathrm a}Z_{\mathbf{a}_k} \ket{\mathrm{e}}\right).
\end{equation*}
We then measure the $k$th stabilizer generator $S_k = X_k Z_{N(k)}$. If the measurement outcome is `$+1$', then the post-measurement state, $\ket{\zeta_+}$, is the appropriately normalized projection $\frac{1}{2}\Big(I+S_k\Big) \ket{\psi'}$. Note that, since the generator matrix $\mathbf{A}$ was chosen to be in systematic form, $S_k$ anti-commutes with $Z_{\mathbf{a}_k}$ and commutes with all other $Z_{\mathbf{a}_i}$'s for $i\neq k$. Using this, we can compute 
\begin{align}\label{general one logical recovered}
   \ket{\zeta_+} &= \sum_{\mathbf{x}' \in \{0,1\}^{k-1}} \Big(a_{\mathbf{x}'} b_0 \ket{0}_{\mathrm a} Z_{\mathbf{x}'\mathbf{A}^\prime}\ket{\mathrm{G}} + a_{\mathbf{x}'} b_1 \ket{1}_{\mathrm a} Z_{\mathbf{x}'\mathbf{A}^\prime}\ket{\mathrm{G}} \Big) \notag \\
        & = \ \ket{\mathrm f} \otimes \left(b_0 \ket{0}_{\mathrm a} + b_1 \ket{1}_{\mathrm a}\right).
\end{align}

On the other hand, if the measurement outcome is `$-1$', then the post-measurement state $\ket{\zeta_-}$ is $\frac12 (I-S_k) \ket{\psi'}$, appropriately normalized. This can be shown to be related to $\ket{\zeta_+}$ as below:
\begin{equation} \label{eq:zeta_-}
    \ket{\zeta_+} = \left(X\otimes Z_{\mathbf{a}_k}\right) \ket{\zeta_-}.
\end{equation}
Thus, the state $\ket{\zeta_-}$ can be transformed to $\ket{\zeta_+}$ by applying $X$ on the {recovery qubit} and $Z_{\mathbf{a}_k}$ on the physical qubits.

In summary, at the end of the recovery protocol, the composite system of physical qubits and {recovery qubit} is as in \eqref{general one logical recovered}. The $k$th logical qubit has thus been extracted into the {recovery qubit} and the remaining $k-1$ logical qubits are left encoded in the state $\ket{\mathrm f} \in \mathcal{Q}'$.

\subsection{Connection between recovery procedure and code distance} \label{subsec:decoding and distance}

In the decoding subsection, we used $CZ$ gates between {recovery qubits} and the physical qubits in the encoded state. We see that the number of $CZ$s used in encoding and recovery are identical. In fact, we apply the exact same tent pegs, and the difference between encoding and recovery only lies in whether we measure the {input qubits} or the graph state qubits. In this section, we show how recovery can involve a smaller number of interactions than what is needed for the encoding if one allows controlled-$X,Y,Z$ gates {controlled by the recovery qubits}.

Since any stabilizer for the graph state contains at least one $X$, strings of $Z$s can never be in the graph stabilizer and, hence, are chosen for constructing CWS codes. This restricts us to using only $CZ$ gates in the encoding procedure. However, this is not the case for the recovery procedure. We can use a mix of controlled gates to retrieve information without any loss.

For instance, consider recovering the single logical qubit $\alpha \ket{0} + \beta \ket{1}$ from its encoding $\ket{\mathrm{e}}= \alpha \ket{\mathrm{G}} + \beta Z_{A} \ket{\mathrm{G}}$. Similar to the recovery procedure discussed in Subsection~\ref{sec:dec}, we append the physical qubits in state $\ket{\mathrm{e}}$ with a qubit in state $\ket{+}_a$. Let $U \in \mathcal{P}_n$ be a Pauli string acting on the $n$ physical qubits. Let $CU$ denote the controlled-$U$ gate controlled by the {recovery qubit}, with the $n$ physical qubits as the target. The action of $CU$ on the joint system results in the state
\begin{equation*}
    \ket{\psi_1} = \frac{1}{2}\left(\ket{0}_a\ket{\mathrm{e}}+\ket{1}_aU\ket{\mathrm{e}}\right).
\end{equation*}

Measuring all the stabilizer generators of $\ket{\mathrm{G}}$ on the physical qubits as done in \eqref{eq:post_proj} gives the state
\begin{equation} \label{U decoding}
   \ket{\psi_2} = \Big[\alpha \ket{0}_a + \Big(\alpha \bra{\mathrm{G}}U \ket{\mathrm{G}}+\beta \bra{\mathrm{G}}U Z_A \ket{\mathrm{G}}\Big)\ket{1}_a\Big]\ket{\mathrm{G}}.
\end{equation}
This expression reduces to the required product state $\left(\alpha \ket{0}_a + \beta \ket{1}_a\right)\ket{\mathrm{G}}$ when the following two conditions are satisfied.
    \begin{equation}\label{UZ}
        \bra{\mathrm{G}}U Z_A \ket{\mathrm{G}} = 1 \ \ \ (\text{i.e., } U Z_A \in \mathcal{S})
    \end{equation}
\begin{equation} \label{U}
    \bra{\mathrm{G}} U \ket{\mathrm{G}} = 0 \ \ \ (\text{i.e., } U \notin \mathcal{S})
\end{equation}

Conditions \ref{UZ} and \ref{U} are not independent, and in fact, we prove the following proposition in Appendix \ref{UZU proof}.
\begin{proposition} \label{UZU_prop}
    If $\mathcal{S}$ is the stabilizer group of a graph state, then for any Pauli operator $U$, 
    $$U Z_A \in \mathcal{S} \implies U \notin \mathcal{S}.$$
\end{proposition}

Therefore, in Eq.~\eqref{U decoding}, the logical state is recovered on the {recovery qubit} when $UZ_A \in \mathcal{S}$. Comparing conditions in Eq.~\eqref{UZ} and \eqref{U} to Theorem~\ref{theorem_1} for one logical qubit, we can see that the support of $U$ or the number of qubits that the {recovery qubit} needs to interact with can be as small as the distance of the code. This can be smaller than the number of $CZ$s used in encoding. 

\section{Concluding Remarks}

In this paper, we revisit the CWS code construction proposed in \cite{CWSCode}. Using highly entangled graph states and carefully chosen classical binary linear codes, we show that $[[n,k,d]]$ CWS codes can be constructed with $d = O(1)$ and $k/n \to 1$ as $n \to \infty$. We further propose a measurement-based encoding algorithm that supports sequential encoding of logical qubits. A recovery algorithm is also given for partial recovery of logical qubits. 

Several interesting directions of future research arise from this work. One is the problem of finding higher-dimensional analogues of the code construction encapsulated within Theorem~\ref{theorem3}. This will require constructing codes whose coordinates are identified with the sites of a $D$-dimensional lattice with periodic boundary conditions (i.e., a $D$-dimensional torus) with certain constraints imposed on the positions of $1$s in codewords. 

The stabilizer group of a CWS code obtained from a graph state and a classical binary linear code is described in Lemma~\ref{codeStabilizers}. It would be of interest to see whether we can impose conditions on the binary linear code so as to obtain stabilizer generators of bounded weight and, hence, quantum low-density parity-check (QLDPC) codes.

Based on practical considerations, it would be useful to explore the possibility of scalably constructing CWS code families with properties such as high error thresholds and fault-tolerant logical operations. It would also be desirable to have fault-tolerant implementations of the measurement-based encoding and recovery procedures we have proposed.

\section*{Acknowledgment}

This work was supported in part by a Prime Minister's Research Fellowship awarded to T.~Aswanth, an IISc-IoE postdoctoral fellowship awarded to N.~Patanker, and by the grant DST/INT/RUS/RSF/P-41/2021 from the Department of Science and Technology, Government of India. We would also like to acknowledge the productive discussions we had with Tania Sidana.

\bibliographystyle{IEEEtran}
\bibliography{bibliofile}

\appendices

\section{Proof of Lemma~\ref{codeStabilizers}}\label{appendix:codeStabilizers}

\begin{proof}
    The states of the form $Z_{\mathbf{c}}\ket{\mathrm{G}}, \mathbf{c} \in \mathcal{C}$, are basis states for the quantum code. Since $\mathcal{S}$ is the stabilizer group of $\ket{\mathrm{G}}$, the stabilizer group of $Z_{\mathbf{c}}\ket{\mathrm{G}}$ is $Z_{\mathbf{c}}\mathcal{S}Z_{\mathbf{c}}$ \cite[Section~10.5.2]{Nielsen_Chuang_2010}. Define $\mathcal{T} := \bigcap_{\mathbf{c} \in \mathcal{C}}{Z_{\mathbf{c}}\mathcal{S}Z_{\mathbf{c}}}$. Since any element in $\mathcal{T}$ stabilizes all $k$ basis states of the quantum code, $\mathcal{T} \subseteq \mathcal{S}_\mathcal{C}$. Similarly, any element in $\mathcal{S}_\mathcal{C}$ must stabilize all the basis states of the quantum code, and therefore, $\mathcal{S}_\mathcal{C} \subseteq \mathcal{T}$. Hence, $\mathcal{S}_\mathcal{C} = \mathcal{T}$.

    Rewrite $\mathcal{T}$ as $\mathcal{S} \cap \left(\bigcap\limits_{\mathbf{0} \neq \mathbf{c} \in \mathcal{C}}{Z_{\mathbf{c}}\mathcal{S}Z_\mathbf{c}}\right)$ and define $\mathcal{T}^\prime := \mathcal{S} \cap \left(\bigcap\limits_{i=1}^{k}{Z_{\mathbf{a}_i}\mathcal{S}Z_{\mathbf{a}_i}}\right)$. Since two Pauli operators either commute or anti-commute, $Z_\mathbf{c}SZ_{\mathbf{c}}$ is either $S$ or $-S$ for all $S\in \mathcal{S}$. This gives an alternative description of $\mathcal{T}$ and $\mathcal{T}^\prime$:
    \begin{equation*}
    \begin{split}
        &\mathcal{T} = \{S \in \mathcal{S}: Z_\mathbf{c}SZ_\mathbf{c} = S \hspace{10pt} \forall \mathbf{c} \in \mathcal{C}\} \hspace{10pt}\text{ and }\\ &\mathcal{T}^\prime = \{S \in \mathcal{S}: Z_{\mathbf{a}_i}SZ_{\mathbf{a}_i} = S \hspace{10pt} \forall i \in [k]\}.
    \end{split}
    \end{equation*}
    Equivalently,
    \begin{equation*}
    \begin{split}
        &\mathcal{T} = \{S \in \mathcal{S}: Z_\mathbf{c}S = SZ_\mathbf{c} \hspace{10pt} \forall \mathbf{c} \in \mathcal{C}\} \hspace{10pt}\text{ and }\\ &\mathcal{T}^\prime = \{S \in \mathcal{S}: Z_{\mathbf{a}_i}S = SZ_{\mathbf{a}_i} \hspace{10pt} \forall i \in [k]\}.
    \end{split}
    \end{equation*}
    Since $\mathbf{a}_i$'s are also codewords in $\mathcal{C}$, by definition, $\mathcal{T} \subseteq \mathcal{T}^\prime$. In the other direction, any element $M \in \mathcal{T}^\prime$ commutes with $Z_{\mathbf{a}_i}$ for all $i \in [k]$, and hence, they commute with $Z_\mathbf{c}$ for all $\mathbf{c} \in \mathcal{C}$. Thus, $\mathcal{T}^\prime \subseteq \mathcal{T}$, and hence, $\mathcal{S}_{\mathcal{C}} = \mathcal{T}=\mathcal{T}^\prime$.
    Defining $\mathbf{a}_0$ to be the zero vector, $\mathcal{S}_\mathcal{C}$ can be written as
    $$
    \mathcal{S}_\mathcal{C} = \bigcap_{i=0}^{k}{Z_{\mathbf{a}_i}\mathcal{S}Z_{\mathbf{a}_i}}.
    $$
\end{proof}

\section{Proof of Theorem \ref{theorem3}}\label{appendix:theorem3}
We have already seen that $d_{\min}(\mathcal{C}_r) \ge 3$. The fact that $\mathcal{C}_r$ does not contain codewords of the form in \eqref{pattern1}--\eqref{pattern4} is proved in Claims~\ref{claim:1}--\ref{claim:4} below. Claims~\ref{claim:1}--\ref{claim:3} do not require the condition that $\Log(1+\alpha) \not\equiv 2 \pmod{3}$, but Claim~\ref{claim:4} does.

\begin{claim}
$\mathcal{C}_r$ contains no codewords of the form in $\eqref{pattern1}$.
\label{claim:1}
\end{claim}
\begin{proof}
Suppose $\mathbf{x}$ is of the form in $\eqref{pattern1}$. For $\mathbf{x}$ to be in $\mathcal{C}$, it must have Hamming weight at least $3$, so at least one of the $\ast$'s in $\mathbf{x}$ is a $1$. However, it is then easily checked that such an $\mathbf{x}$ cannot be in the nullspace of the second row of the $H$ matrix.
\end{proof}

\begin{claim}
$\mathcal{C}_r$ contains no codewords of the form in $\eqref{pattern2}$.
\label{claim:2}
\end{claim}
\begin{proof}
Suppose $\mathbf{x}$ is of the form in $\eqref{pattern2}$. Now, it can be checked that such an $\mathbf{x}$ is in the nullspace of the second row of the $H$ matrix iff both the $\ast$'s in $\mathbf{x}$ are $1$s. In this case, we would have ${H \mathbf{x}^T = [\alpha^i+\alpha^{i+1}+\alpha^{i+3}+\alpha^{i+4}, \ \ 0]^T}$. So, for $\mathbf{x}$ to be in $\mathcal{C}_r$, we would need ${\alpha^i+\alpha^{i+1}+\alpha^{i+3}+\alpha^{i+4} = 0}$, or equivalently, $\alpha^i(1+\alpha)(1+\alpha^3) = 0$. But this is impossible as $\alpha$ is a primitive element of $\mathbb{F}_{2^{2r}}$ with $r \ge 2$.
\end{proof}

\begin{claim}
$\mathcal{C}_r$ contains no codewords of the form in $\eqref{pattern3}$.
\label{claim:3}
\end{claim}
\begin{proof}
Suppose $\mathbf{x}$ is of the form in $\eqref{pattern3}$. For such an $\mathbf{x}$ to be in the nullspace of the second row of the $H$ matrix, either both the $\ast$'s in $\mathbf{x}$ are $0$s or both are $1$s. In either case, we obtain that $H\mathbf{x}^T$ is of the form $[\alpha^i+c \, \alpha^{i+1} + \alpha^{i+2}+\alpha^{i+3+j}+c \, \alpha^{i+3+j+1} + \alpha^{i+3+j+2}, \ \ 0]^T$, for some $c \in \{0,1\}$ and $i,j \ge 0$. (Here, $j$ is the length of the run of $0$s between the two `$1 \ast 1$' blocks, so $j \leq n-6$.) So, for $\mathbf{x}$ to be in $\mathcal{C}_r$, we would need $\alpha^i+c \, \alpha^{i+1} + \alpha^{i+2}+\alpha^{i+3+j}+c \, \alpha^{i+3+j+1} + \alpha^{i+3+j+2} = 0$, or equivalently, $\alpha^i(1+c \, \alpha + \alpha^2)(1+\alpha^{j+3}) = 0$. By primitivity of $\alpha$, we cannot have $\alpha^i = 0$ or $1+\alpha^{j+3} = 0$ (note that $j+3 < 2^{2r}-1$). Moreover, $1 + c \, \alpha + \alpha^2 = 0$ would mean that $\alpha$ is a root of a degree-$2$ polynomial in $\mathbb{F}_2[x]$, which is impossible, as the minimal polynomial of $\alpha$ has degree $2r > 2$. We thus conclude that we cannot have $\alpha^i(1+c \, \alpha + \alpha^2)(1+\alpha^{j+3}) = 0$, and so no $\mathbf{x}$ of the form in \eqref{pattern3} can be in the code $\mathcal{C}_r$.
\end{proof}

\begin{claim}
If \, $\Log(1 + \alpha) \not\equiv 2 \pmod{3}$, then $\mathcal{C}_r$ contains no codewords of the form in $\eqref{pattern4}$.
\label{claim:4}
\end{claim}
\begin{proof}
Suppose $\mathbf{x}$ is of the form in $\eqref{pattern4}$. Such an $\mathbf{x}$ can be in the nullspace of the second row of the $H$ matrix only if exactly one of the two $\ast$'s in $\mathbf{x}$ is a $1$. In other words, an $\mathbf{x}$ of the form in \eqref{pattern4} is in $\mathcal{C}_r$ only if $\mathbf{x} = 0 \cdots 0 1 1 1 0 \cdots 0$   or  $\mathbf{x} = 0 \cdots 0 1 0 1 0 \cdots 0 1 0 \cdots 0$. (Here, the $0 \cdots 0$ blocks represent runs, possibly empty, of $0$s.) 

If $\mathbf{x} = 0 \cdots 0 1 1 1 0 \cdots 0$, then $H\mathbf{x}^T = [\alpha^i(1+\alpha+\alpha^2), \ \ 0]^T$ for some $i$. So, for this $\mathbf{x}$ to be in $\mathcal{C}_r$, we need $1+\alpha+\alpha^2 = 0$, which is impossible since $\alpha$ is a primitive element of $\mathbb{F}_{2^{2r}}$, so its minimal polynomial over $\mathbb{F}_2$ has degree $2r > 2$.

When $\mathbf{x} = 0 \cdots 0 1 0 1 0 \cdots 0 1 0 \cdots 0$, we have $H \mathbf{x}^T = [\alpha^i(1+\alpha^2+\alpha^j), \ \ \alpha^{bi}(1+\alpha^{2b} + \alpha^{bj})]^T$ for some $i \ge 0$ and $j \ge 3$. For the second entry in $H\mathbf{x}^T$ to be $0$, we need ${(\alpha^{b})}^j = \alpha^b$, which (by choice of $b$) happens iff $j \equiv 1 \pmod{3}$. For the first entry in $H\mathbf{x}^T$ to also be $0$, we further need $1+\alpha^2 = \alpha^j$. Thus, we can have $\mathbf{x} \in \mathcal{C}_r$ (i.e., $H\mathbf{x}^T = [0,\ 0]^T$) iff $1+\alpha^2 = \alpha^j$ for some $j \equiv 1\pmod{3}$. In other words, $\mathbf{x} \in \mathcal{C}_r$ iff $\Log(1+\alpha^2) \equiv 1 \pmod{3}$. We next show that the latter condition is equivalent to $\Log(1+\alpha) \equiv 2 \pmod{3}$, which will suffice to prove the claim.

Suppose $\Log(1+\alpha) = \ell$, meaning that $\alpha^{\ell} = 1+\alpha$. By squaring both sides, we have $\alpha^{2\ell \mod (4^r-1)} = 1+\alpha^2$, i.e., $\Log(1+\alpha^2) = 2\ell \mod(4^r-1)$. Since $4^r-1$ is a multiple of $3$, we have ${2\ell \mod(4^r-1) \equiv 1 \pmod{3}}$ iff $2\ell \equiv 1 \pmod{3}$, which in turn holds iff ${\ell \equiv 2 \pmod{3}}$. In other words, $\Log(1+\alpha^2) \equiv 1 \pmod{3}$ iff $\Log(1+\alpha) \equiv 2 \pmod{3}$.
\end{proof}

This completes the proof of the Theorem.

\section{ {Proof of Proposition~\ref{resolved_conjecture}}} \label{conjecture_proof}

We first note that the proposition holds for $r=2$. Indeed, we know that there exists a primitive element $\alpha \in \mathbb{F}_{2^4}$ such that $1+\alpha=\alpha^4$ (as $1+x+x^4 \in \mathbb{F}_{2}[x]$ is a primitive polynomial). For this $\alpha$, we have $\Log(1+\alpha) \equiv 1 \pmod{3}$.

The aim of the rest of this appendix is to prove the following sharpening of Proposition~\ref{resolved_conjecture}, applicable for $r \ge 3$.

\begin{proposition} \label{prop:sharpening}
    For any integer $r\geq 3$, there exists a primitive element $\alpha \in \mathbb{F}_{2^{2r}}$ such that $\Log(1+\alpha) \equiv 0 \pmod 3$.    
\end{proposition}

We first provide an outline of our proof technique. Suppose there exists $z \in \mathbb{F}_{2^{2r}}$, $r \geq 3$, such that $1+z^3$ is a primitive element of $\mathbb{F}_{2^{2r}}$. Let us call this primitive element $y$, i.e.\ $1+z^3=y$. Then $z=y^s$ for some $0< s \leq 2^{2r}-2$. So, we get $1+y+y^{3s}=0$ which implies that $\Log(1+y)\equiv 0 \pmod 3$. Thus, to prove Proposition~\ref{prop:sharpening}, it is enough to check whether there exists a $z \in \mathbb{F}_{2^{2r}}$ such that $1+z^3$ is a primitive. To this end, let $q = 2^{2r}$ and $f(x)=1+x^3 \in \mathbb{F}_q[x]$. Let $N_f(q-1)$ denote the number of (non-zero) $z \in \mathbb{F}_q$ such that $f(z)$ is primitive. Our aim is to show that $N_f(q-1) > 0$. 

We give some definitions and results that will be required in our proof. For details, we refer the reader to \cite{finitefields}.

Let $G$ be a ﬁnite abelian group. A character $\chi$ of $G$ is a homomorphism from $G$ into the multiplicative group $U$ of complex numbers of absolute value $1$. The set $\widehat{G}$ of all characters of $G$ forms a group under multiplication and is isomorphic to $G$. The character $\chi_1$, defined as $\chi_1(g)=1~\forall ~g \in G$, is the trivial character of $G$.\par
For a finite field $\mathbb{F}_q$, the characters of the additive group $\mathbb{F}_q$ are called additive characters, and the characters of the multiplicative group $\mathbb{F}_q^{*}$ are called multiplicative characters. Multiplicative characters are extended to zero using the rule,
\begin{equation} \label{eq:chi0}
    \chi(0)=\begin{cases}
   0 \text{~if~} \chi \neq \chi_1,\\
   1 \text{~if~} \chi=\chi_1.
\end{cases}
\end{equation}
The set of multiplicative characters, $\widehat{\mathbb{F}_q^{*}}$, is cyclic. Let $\chi_d$ denote a multiplicative character of order $d$ for any $d|(q-1)$. There are $\phi(d)$ such characters in $\widehat{\mathbb{F}_{q}^{*}}$, where $\phi$ is Euler's totient function. From \cite{finitefields}, we have the following result: For $a \in \mathbb{F}_{q}^{*}$,
\begin{equation} \label{eq:primitive}
\frac{\phi(q-1)}{(q-1)} \sum_{d|(q-1)} \frac{\mu(d)}{\phi(d)} \sum_{\chi_d} \chi_d(a)=\begin{cases}
    1 &\text{~if~} a \text{~is primitive},\\
    0 &\text{~otherwise},
\end{cases}
\end{equation}
where the outer sum $d$ runs over all positive divisors of $(q-1)$, and in the inner sum, $\chi_d$ runs through the $\phi(d)$ multiplicative characters of $\mathbb{F}_q$ of order $d$. Here, $\mu$ denotes the M\"{o}bius function on $\mathbb{N}$ defined by 
\[ \mu(n)=
\begin{cases}
 1 &\text{~if~} n=1,\\
 (-1)^k &\text{~if~} n \text{~is the product of } k \text{~distinct primes},\\
 0 &\text{~if~} n \text{ is divisible by a square of prime}.
\end{cases}
\]

\begin{theorem}\cite[Theorem~5.41]{finitefields} \label{conjectureProof_theorem1}
Let $\chi$ be a multiplicative character of $\mathbb{F}_q$ of order $m>1$, and let $f \in \mathbb{F}_q[x]$ be a monic polynomial of positive degree such that $f(x)$ is not of the
form $g(x)^m$, where $g(x) \in \mathbb{F}_q[x]$ with degree at least 1. Suppose $\ell$ is the number of distinct roots of $f$ in its splitting ﬁeld over $\mathbb{F}_q$. Then for every $a \in \mathbb{F}_q$, we have $$\sum_{c \in \mathbb{F}_q} |\chi(a(f(c))|\leq (\ell-1) q^{1/2}.$$
\end{theorem}

Recall that our aim is to show that $N_f(q-1) > 0$ where, for $q = 2^{2r}$ and $f(x) = 1+x^3$, $N_f(q-1)$ is the number of $z \in \mathbb{F}_{2^{2r}}$ such that $f(z)$ is primitive. For $n \in \mathbb{N}$, let $w(n)$ denote the number of distinct prime divisors of $n$. 

\begin{claim} \label{claim_conjecture}
    For $r \ge 3$, if $r>w(q-1)$, then $N_f(q-1)>0$.
\end{claim}
\begin{proof} For $z \in \mathbb{F}_q \setminus \{1, \omega, \omega^2\}$, where $\omega \in \mathbb{F}_q$ is a primitive cube root of unity, we have $f(z) \ne 0$. We then have, via \eqref{eq:primitive}, 
    \begin{align*}
    N_f & (q-1) \\ 
    &=\sum_{z \in \mathbb{F}_q \setminus \{1, \omega, \omega^2\}} \frac{\phi(q-1)}{(q-1)} \sum_{d|(q-1)} \frac{\mu(d)}{\phi(d)} \sum_{\chi_d} \chi_d(f(z))\\
    &= \frac{\phi(q-1)}{(q-1)} \sum_{d|(q-1)} \frac{\mu(d)}{\phi(d)} \sum_{\chi_d} \sum_{z \in \mathbb{F}_{q} \setminus \{1, \omega, \omega^2\}} \chi_d(f(z))
\end{align*}
The contribution of $d=1$ to the sum over the divisors of $q-1$ is $q-3$. To determine the contribution of divisors $d \ne 1$, observe that, for any $\chi_d$ with $d \ne 1$, we have $\chi_d(0) = 0$ (see \eqref{eq:chi0}), and hence,
$$
\sum_{z \in \mathbb{F}_{q} \setminus \{1, \omega, \omega^2\}} \chi_d(f(z)) = \sum_{z \in \mathbb{F}_{q}} \chi_d(f(z)).
$$
Moreover, $f(x) = 1+x^3 = (x+1)(x+\omega)(x+\omega^2)$ is not of the form $g(x)^d$ for any $d > 1$. Therefore, 
\begin{align*}
    \bigg|\sum_{1 \neq d|(q-1)} \frac{\mu(d)}{\phi(d)} \sum_{\chi_d} & \sum_{z \in \mathbb{F}_{q} \setminus \{1, \omega, \omega^2\}} \chi_d(f(z))\bigg| \\ 
    & = \bigg|\sum_{1 \neq d|(q-1)} \frac{\mu(d)}{\phi(d)} \sum_{\chi_d} \sum_{z \in \mathbb{F}_{q} } \chi_d(f(z))\bigg| \\ 
    & \leq \ \sum_{1 \neq d|(q-1)} \frac{|\mu(d)|}{\phi(d)} \sum_{\chi_d} \bigg| \sum_{z \in \mathbb{F}_{q} } \chi_d(f(z))\bigg| \\
    & \stackrel{\mathrm{(a)}}{\leq} \ 2 q^{1/2}\sum_{1 \neq d|(q-1)} \frac{|\mu(d)|}{\phi(d)} \sum_{\chi_d} 1 \\
    &= \ 2q^{1/2} \sum_{1 \neq d|(q-1)} |\mu(d)|\\
    &=\ 2q^{1/2}[W(q-1)-1]
\end{align*}
where $W(q-1)$ denotes the number of distinct square-free divisors of $q-1$. The inequality~(a) above is obtained via Theorem~\ref{conjectureProof_theorem1}. 

Thus,
\begin{align*}
N_{f} & (q-1) \\
& = \frac{\phi(q-1)}{(q-1)}\Big[q-3+\sum_{1 \neq d|(q-1)} \frac{\mu(d)}{\phi(d)} \sum_{\chi_d} \sum_{z \in \mathbb{F}_{q}} \chi_d(f(z))\Big] \\& \geq \frac{\phi(q-1)}{q-1}\bigl[q-3-2 q^{1/2}\bigl(W(q-1)-1\bigr)\bigr] \\
& \geq \frac{\phi(q-1)}{q-1}\bigl[q-2 q^{1/2} W(q-1) + 2q^{1/2} - 3\bigr].
\end{align*}
Now, $2q^{1/2} - 3 = 2^{r+1} - 3 > 0$ for $r \ge 3$. Hence, if $q^{1/2}\geq 2W(q-1)$, i.e., $2^{r-1}\geq W(2^{2r}-1)$, then $N_f(q-1)>0$. Note that $W(q-1)=2^{w(q-1)}$, recalling that $w(q-1)$ is the number of distinct prime factors of $q-1$.
\end{proof}

With Claim~\ref{claim_conjecture} in hand, all we need to complete the proof of Proposition~\ref{prop:sharpening} is to show that for all $r \ge 3$, the condition $r > w(2^{2r}-1)$ holds.

    If, for some odd integer $n$, we have $\omega(n)\geq r$, then 
    $n \geq p_1p_2 \cdots p_r$, where $p_i,~1\leq i \leq r$ denotes the $i$-th odd prime. We have $p_1 p_2 p_3=3 \times 5 \times 7 =105> 4^3$ and $p_i> 4$ for $i \geq 4$. So, for $r \ge 3$, we have 
    $$n\geq p_1 p_2 \cdots p_r> 4^r.$$
    Consequently, since $2^{2r}-1 < 4^r$, we must have $w(2^{2r}-1)<r$. 
    This completes the proof of Proposition~\ref{prop:sharpening}, and hence, that of Proposition~\ref{resolved_conjecture}.

\section{Proof of Theorem~\ref{thm:2D}} \label{appendix:thm2D}

The goal is to verify the condition in Theorem~\ref{theorem_1}, namely, that (in the notation of that theorem) the minimum symplectic weight of a non-zero vector in $\mathcal{W} := \mathcal{V}_{\mathcal S} + \mathcal{C}'$ is $2D+1 = 5$. Here, $\mathcal{C}'$ is the length-$2n^2$ extension of $C_U$, obtained by prefixing $n^2$ $0$'s to each codeword $\mathbf{c} \in C_U$. To do this, we first check that $d_{\min}(C_U) \ge 5$, so that the symplectic weight of any vector in $\mathcal{C}'$ is at least $5$. Then, we verify that adding a binary vector of symplectic weight less than $5$ to the symplectic representation of a non-identity stabilizer in $\mathcal{S}$ does not yield a codeword in $\mathcal{C}'$; this shows that adding a non-zero vector in $\mathcal{V}_{\mathcal S}$ to a codeword in $\mathcal{C}'$ does not result in a vector of symplectic weight less than $5$.

For the purposes of this proof, we define the weight of a polynomial $f(x,y) = a_{i,j}x^iy^j \in \mathbb{F}_2[x,y]$ as the number of nonzero coefficients $a_{i,j}$. 

\begin{claim}
    The minimum distance of $C_U$ is at least $6$.
\end{claim}
\begin{proof}
    We will show that $d(C_U) \geq 5$. Then, since $(1,1) \in U$, the weight of any $f(x,y) \in C_U$ must be even, so that $d(C) \geq 6$ must in fact hold.

    Consider a polynomial $f(x,y) = l_1 x^{a_1}y^{b_1}+l_2 x^{a_2}y^{b_2}+l_3 x^{a_3}y^{b_3}+l_4 x^{a_4}y^{b_4} \in \mathbb{F}_2[x,y]$, where $l_i \in \{0,1\}$ for $i=1,2,3,4$, and $(a_i,b_i) \ne (a_j,b_j)$ for $i \ne j$. If an odd number $l_i$'s are non-zero, then $f(1,1) = 1 \ne 0$, so that $f(x,y) \notin C_U$. So, henceforth, we assume that an even number of $l_i$'s are non-zero. 
    
    If $f(x,y) \in C_U$, then from the fact that $f(\alpha,\alpha) = f(\alpha,1) = f(\alpha^3,1) = f(1,\alpha) =f(1,\alpha^3) = 0$, we obtain
	\begin{align}
 l_1 \alpha^{a_1+b_1}+l_2 \alpha^{a_2+b_2}+l_3 \alpha^{a_3+b_3}+l_4 \alpha^{a_4+b_4} \ &= \ 0 \label{appD:eq1} \\
 l_1 \alpha^{a_1}+l_2 \alpha^{a_2}+l_3 \alpha^{a_3}+l_4 \alpha^{a_4} &= \ 0 \label{appD:eq2} \\
 l_1 \alpha^{3 a_1}+l_2 \alpha^{3 a_2}+l_3 \alpha^{3 a_3}+l_4 \alpha^{3 a_4} &= \ 0 \label{appD:eq3} \\
 l_1 \alpha^{b_1}+l_2 \alpha^{b_2}+l_3 \alpha^{b_3}+l_4 \alpha^{b_4} &= \ 0 \label{appD:eq4} \\
 l_1 \alpha^{3 b_1}+l_2 \alpha^{3 b_2}+l_3 \alpha^{3 b_3}+l_4 \alpha^{3 b_4} &= \ 0 \label{appD:eq5}
 \end{align}
	If the $a_i$'s are all distinct, then using \eqref{appD:eq2} and \eqref{appD:eq3}, we can write
	\begin{align*}
		\left[ \begin{array}{cccc}
		\alpha^{a_1} &\alpha^{a_2} &\alpha^{a_3} &\alpha^{a_4}\\
		\alpha^{2a_1} &\alpha^{2a_2} &\alpha^{2a_3} &\alpha^{2a_4}\\
		\alpha^{3a_1} &\alpha^{3a_2} &\alpha^{3a_3} &\alpha^{3a_4}\\
		\alpha^{4a_1} &\alpha^{4a_2} &\alpha^{4a_3} &\alpha^{4a_4}\\
		\end{array} \right]
		\left[ \begin{array}{c}
			l_1\\
			l_2\\
			l_3\\
			l_4\\
	\end{array} \right]
=	\left[ \begin{array}{c}
	0\\
	0\\
	0\\
	0\\
\end{array} \right]
	\end{align*}
The Vandermonde matrix above has nonzero determinant, so that the only solution of the above system of equations is $(l_1,l_2,l_3,l_4)=(0,0,0,0)$, i.e., $f(x,y) = 0$. Similarly, if the $b_i$'s are all distinct, then an analogous argument using \eqref{appD:eq4} and \eqref{appD:eq5} leads to the same conclusion. 

Let us now assume, without loss of generality, that $a_1=a_2$ and some of the $b_i$'s are possibly equal. Then, from \eqref{appD:eq2}, we have
\begin{equation}
    (l_1+l_2)\alpha^{a_1}+l_3\alpha^{a_3}+l_4 \alpha^{a_4}=0. \label{appD:eq6}
\end{equation}
We now consider the various possibilities for $(l_1,l_2,l_3,l_4)$ with an even number of $l_i$'s being non-zero:
\begin{itemize}
\item If $(l_1,l_2,l_3,l_4)=(0,0,1,1)$, then from \eqref{appD:eq6} we get $\alpha^{a_3}=\alpha^{a_4}$. This implies $a_3=a_4$. Also, from \eqref{appD:eq4}, we get $\alpha^{b_3}=\alpha^{b_4}$. Hence we get $(a_3,b_3)=(a_4,b_4)$, a contradiction.
\item  If $(l_1,l_2,l_3,l_4)=(0,1,0,1)$, then from \eqref{appD:eq6} we get $\alpha^{a_1}=\alpha^{a_4}$. This implies $a_1=a_2=a_4$. Also, from \eqref{appD:eq4}, we get  $\alpha^{b_2}=\alpha^{b_4}$. Hence we get $(a_2,b_2)=(a_4,b_4)$, a contradiction.
\item If $(l_1,l_2,l_3,l_4)=(0,1,1,0)$, then from \eqref{appD:eq6}, we get $\alpha^{a_1}=\alpha^{a_3}$. This implies $a_1=a_3=a_2$. Also, from equation \eqref{appD:eq4}, we get $\alpha^{b_2}=\alpha^{b_3}$. Hence we get $(a_2,b_2)=(a_3,b_3)$, a contradiction.

 \item If $(l_1,l_2,l_3,l_4)=(1,0,0,1)$, then $\alpha^{a_1}=\alpha^{a_4}$. This implies $a_1=a_2=a_4$. Also,  $\alpha^{b_1}=\alpha^{b_4}$. Hence we get $(a_1,b_1)=(a_4,b_4)$, a contradiction.
 \item If $(l_1,l_2,l_3,l_4)=(1,0,1,0)$, then similar to the previous case, we get $(a_1,a_3)=(b_1,b_3)$.
 \item If $(l_1,l_2,l_3,l_4)=(1,1,0,0)$, then \eqref{appD:eq4} and \eqref{appD:eq1}  gives $b_1=b_2$. Hence, $(a_1,b_1)=(a_2,b_2)$, a contradiction.
 \item If $(l_1,l_2,l_3,l_4)=(1,1,1,1)$ then we get, from \eqref{appD:eq6}, $a_3=a_4$. With this, \eqref{appD:eq1}  gives
	 $$\alpha^{a_1}(\alpha^{b_1}+\alpha^{b_2})+\alpha^{a_3}(\alpha^{b_3}+\alpha^{b_4})=0.$$
	 Using \eqref{appD:eq4}, we then get $(\alpha^{b_1}+\alpha^{b_2})(\alpha^{a_1}+\alpha^{a_3})=0$. Now if $b_1=b_2$ then $(a_1,b_1)=(a_2,b_2)$. If $a_1=a_3$, then we have $a_1=a_2=a_3=a_4$ and we have already assumed not all the $b_i$'s are distinct. Thus, this case too ends in a dead end.
\end{itemize}
Thus, all cases end in a contradiction, which proves that no polynomial of weight $4$ or less is in $C_U$. We conclude that $d_{\min}(C_U) \ge 5$.
\end{proof}

We would now like to show that adding a binary vector of symplectic weight less than $5$ to the symplectic representation of a non-identity stabilizer in $\mathcal{S}$ does not yield a codeword in $\mathcal{C}'$. We introduce some convenient notation first. For a subset $J$ of the vertex set, $V := {[n]}^2$, of the lattice $\Lambda_{n,n}$, let $\mathbf{e}^J = (e^J_i: i \in V)$ denote the vector in $\{0,1\}^V$ with $e^J_i = 1$ iff $i \in J$. When $J = \{i\}$ is a singleton, we simplify the notation to $\mathbf{e}^i$. With this, the symplectic representation of a stabilizer generator $S_i = X_i Z_{N(i)}$ is $\mathbf{s}_i := [\mathbf{e}^i \mid \mathbf{e}^{N(i)}]$. Also, recall that a codeword in $\mathcal{C}'$ is of the form $[\mathbf{0} \mid \mathbf{c}]$ for some $\mathbf{c} \in C_U$, where $\mathbf{0}$ denotes the zero vector in $\{0,1\}^V$. 

The \emph{support} of a vector $\mathbf{x} = (x_i)_{i \in V} \in \{0,1\}^V$ is defined as $\text{supp}(\mathbf{x}) := |\{i: x_i = 1\}|$. We further define the \emph{$X$-support} and \emph{$Z$-support} of a vector $\mathbf{v} = [\mathbf{v}^{(1)} \mid \mathbf{v}^{(2)}] \in \{0,1\}^V \times \{0,1\}^V$ to be $\text{supp}_X(\mathbf{v}) := \text{supp}(\mathbf{v}^{(1)})$ and $\text{supp}_Z(\mathbf{v}) := \text{supp}(\mathbf{v}^{(2)})$, respectively. The symplectic weight of $\mathbf{v}$ then is $\text{wt}_{\mathrm{s}}(\mathbf{v}) = |\text{supp}_X(\mathbf{v}) \cup \text{supp}_Z(\mathbf{v})|$.

If we have a stabilizer whose symplectic representation, $\mathbf{s}$, is obtained by adding five or more of the generators $[\mathbf{e}^i \mid \mathbf{e}^{N(i)}]$, then $|\text{supp}_X(\mathbf{s})| \ge 5$. So, we cannot get a vector of the form $[\mathbf{0} \mid \mathbf{f}]$ by adding a vector of symplectic weight less than $5$ to $\mathbf{s}$. Therefore, we need only consider stabilizers with symplectic representation $\mathbf{s}$ obtained as the sum of four or fewer generators $\mathbf{s}_i = [\mathbf{e}^i \mid \mathbf{e}^{N(i)}]$. To bring such an $\mathbf{s}$ into the form $[\mathbf{0} \mid \mathbf{f}]$ by adding a vector $\mathbf{b} = [\mathbf{b}^{(1)} \mid \mathbf{b}^{(2)}]$ of symplectic weight less than $5$, we need $\text{supp}_X(\mathbf{b}) = \text{supp}_X(\mathbf{s})$ and $|\text{supp}_X(\mathbf{b}) \cup \text{supp}_Z(\mathbf{b})| \le 4$.
For the various possible choices of $\mathbf{s}$ and $\mathbf{b}$ such that $\mathbf{s}+\mathbf{b}$ is of the form $[\mathbf{0} \mid \mathbf{f}]$, the vector $\mathbf{f}$, when represented as a polynomial in $\mathbb{F}_2[x,y]$, takes one of the forms listed below.
\begin{itemize}
\item Type I (when $\text{supp}_Z(\mathbf{b}) \subseteq \text{supp}_X(\mathbf{b}) = \text{supp}_X(\mathbf{s})$): 
\begin{align*}
f(x,y) = \ & p_1 x^{a_1}y^{b_1}(x+y+Axy+xy^2+x^2 y)\\
&+ p_2 x^{a_2}y^{b_2}(x+y+Bxy+xy^2+x^2 y) \\ 
&+ p_3 x^{a_3}y^{b_3}(x+y+Cxy+xy^2+x^2 y)\\
&+ p_4 x^{a_4}y^{b_4}(x+y+Dxy+xy^2+x^2 y)
\end{align*}
where $A,B,C,D,p_1,p_2,p_3,p_4 \in \{0,1\}$, with the $p_i$'s not all $0$, and $(a_i,b_i),~i=1,2,3,4$, are distinct.
\item Type II (when $|\text{supp}_X(\mathbf{s})| = 3$, $\text{supp}_X(\mathbf{b}) = \text{supp}_X(\mathbf{s})$, and $|\text{supp}_Z(\mathbf{b}) \setminus \text{supp}_X(\mathbf{b})| = 1$):
\begin{align*}
    f(x,y)= &x^{a_1} y^{b_1}(x+y+Axy+x y^2+x^2 y)\\
    &+ x^{a_2} y^{b_2}(x+y+Bxy+x y^2+x^2 y) \\ 
    &+ x^{a_3} y^{b_3}(x+y+Cxy+x y^2+x^2 y)\\
    &+ x^{a_4} y^{b_4},
\end{align*}
where $A,B,C,D \in \{0,1\}$, and $(a_i,b_i),~i=1,2,3$, are distinct.
\item Type III (when $|\text{supp}_X(\mathbf{s})| = 2$, $\text{supp}_X(\mathbf{b}) = \text{supp}_X(\mathbf{s})$, and $|\text{supp}_Z(\mathbf{b}) \setminus \text{supp}_X(\mathbf{b})| \in \{1,2\}$):
\begin{align*}
    f(x,y)= &x^{a_1} y^{b_1}(x+y+Axy+x y^2+x^2 y)\\
    &+ x^{a_2} y^{b_2}(x+y+Bxy+x y^2+x^2 y) \\
    &+ p_3 x^{a_3} y^{b_3} + p_4 x^{a_4} y^{b_4},
\end{align*}
where $A, B \in \{0,1\}$, $p_3,p_4 \in \{0,1\}$ but not both $0$, $(a_1,b_2) \ne (a_2,b_2)$ and $(a_3,b_3) \ne (a_4,b_4)$.
\item Type IV (when $|\text{supp}_X(\mathbf{s})| = 1$, $\text{supp}_X(\mathbf{b}) = \text{supp}_X(\mathbf{s})$, and $|\text{supp}_Z(\mathbf{b}) \setminus \text{supp}_X(\mathbf{b})| \in \{1,2,3\}$):
\begin{align*}
f(x,y)= &x^{a_1} y^{b_1}(x+y+Axy+x y^2+x^2 y)+p_2 x^{a_2} y^{b_2}\\
&+ p_3 x^{a_3} y^{b_3}+p_4 x^{a_4}y^{b_4},
\end{align*}
where $A \in \{0,1\}$, $p_2,p_3,p_4 \in \{0,1\}$ but not all $0$, and $(a_i,b_i)$, $i=2,3,4$, are distinct.
\end{itemize}

We undertake a case-by-case analysis in Claims~\ref{claim:Type-I}--\ref{claim:Type-IV} below to show that none of the above four types of polynomials can be in $C_U$. 

\begin{claim}
    Polynomials of Type I are not in $C_U$.
    \label{claim:Type-I}
\end{claim}
\begin{proof}
Let $f(x,y)$ be a polynomial of Type I.
\begin{enumerate}
    \item Suppose that exactly one of the $p_i$'s is non-zero; without loss of generality (WLOG), $p_1 = 1$ and $p_2=p_3=p_4=0$. Then, we have $f(x,y)=x^{a_1}y^{b_1}(x+y+Axy+xy^2+x^2 y)$, which has weight at most $5$, hence cannot be in $C_U$ since $d_{\min}(C_U) \ge 6$.
    \item Suppose that exactly two of the $p_i$'s are non-zero; WLOG, $p_1 = p_2 = 1$ and $p_3=p_4=0$. Then, we have $f(x,y)=x^{a_1}y^{b_1}(x+y+Axy+xy^2+x^2 y)+x^{a_2}y^{b_2}(x+y+Bxy+xy^2+x^2 y)$. 
    \begin{itemize}
        \item[-] Assume $(A,B)=(0,0)$. If $f(x,y) \in C_U$, then $f(1,\alpha)=(\alpha^{b_1}+\alpha^{b_2})(1+\alpha^2)=0$, which implies $b_1=b_2$. Similarly, $f(\alpha,1)=(\alpha^{a_1}+\alpha^{a_2})(1+\alpha^2)=0$, which implies $a_1=a_2$. So $(a_1,b_1)=(a_2,b_2)$, a contradiction. 
        \item[-] Assume that exactly one of $A$ or $B$ is $1$; WLOG, $A=1$ and $B=0$. Then, $f(1,1)=1 \neq 0$, which shows that $f(x,y) \notin C_U$. 
        \item[-] Assume that $(A,B)=(1,1)$. If $f(x,y) \in C_U$, we have $f(1,\alpha)=(\alpha^{b_1}+\alpha^{b_2})(1+\alpha+\alpha^2)=0$, which implies $b_1=b_2$. Similarly $f(\alpha,1)=(\alpha^{a_1}+\alpha^{a_2})(1+\alpha+\alpha_2)=0$, which implies $a_1=a_2$. So $(a_1,b_1)=(a_2,b_2)$, a contradiction.
    \end{itemize}
    \item Suppose that exactly three of the $p_i$'s are non-zero; WLOG, $p_1 = p_2 = p_3 = 1$ and $p_4=0$. Then, we have $f(x,y)=x^{a_1}y^{b_1}(x+y+Axy+xy^2+x^2 y)+x^{a_2}y^{b_2}(x+y+Bxy+xy^2+x^2 y)+x^{a_3}y^{b_3}(x+y+Cxy+xy^2+x^2 y)$. 
    \begin{itemize}
    \item[-] Assume $(A,B,C)=(0,0,0)$. If $f(x,y) \in C_U$, then $f(1,\alpha)=0$ implies $\alpha^{b_1}+\alpha^{b_2}+\alpha^{b_3}=0$, and $f(1,\alpha^3)=0 $ implies $\alpha^{3 b_1}+\alpha^{3 b_2}+\alpha^{3 b_3}=0$. This gives
    \begin{align*}
        & \alpha^{3 b_1}+\alpha^{3 b_2}=\alpha^{3 b_3}=(\alpha^{b_1}+\alpha^{b_2})^3\\
        \implies & \alpha^{b_1+b_2}(\alpha^{b_1}+\alpha^{b_2})=0\\
        \implies & b_1=b_2.
        \end{align*}
    Similarly we can get $a_1=a_2$. So $(a_1,b_1)=(a_2,b_2)$, which is a contradiction.
    \item[-] Assume that an odd number of the parameters $A$, $B$, $C$ is non-zero. Then,  $f(1,1)=1 \neq 0$, which implies $f(x,y) \notin C_U$. 
    \item[-] Assume that exactly two of $A,~B,~C$ are non-zero; say, $A=B=1$ and $C = 0$. Then, $f(1, \gamma)=(\gamma^{b_1}+\gamma^{b_2})(1+\gamma+\gamma^2)+\gamma^{b_3} (1+\gamma^2)=\gamma^{b_3+1} \neq 0$. Hence $f(x,y) \notin C_U$.
    \end{itemize}
    \item Suppose that $p_1=p_2=p_3=p_4=1$. Then, $f(x,y)=x^{a_1}y^{b_1}(x+y+Axy+xy^2+x^2 y)+x^{a_2}y^{b_2}(x+y+Bxy+xy^2+x^2 y)+x^{a_3}y^{b_3}(x+y+Cxy+xy^2+x^2 y)+x^{a_4}y^{b_4}(x+y+Dxy+xy^2+x^2 y)$. 
    \begin{itemize}
    \item[-] Assume that $(A,B,C,D)=(0,0,0,0)$. If $f(x,y) \in C_U$, then $f(1,\alpha)=0$ and $f(1,\alpha^3)=0$ implies
    \begin{equation}
    \begin{split}
    &\alpha^{b_1}+\alpha^{b_2}+\alpha^{b_3}+\alpha^{b_4}=0 \text{~ and~}\\
    &\alpha^{3 b_1}+\alpha^{3 b_2}+\alpha^{3 b_3}+\alpha^{3 b_4}=0. \label{eq:S1}
    \end{split}
    \end{equation}
    If the $b_i$'s are all distinct, then as argued in the proof of $d_{\min}(C_U) \ge 6$, such a system of equations doesn't have a solution. Similarly, if the $a_i$'s are all distinct, then an analogous argument leads to the same conclusion. So, we assume WLOG that $b_1=b_2$ (so that $b_3=b_4$ as well, by Eq.~\eqref{eq:S1}), and also that not all $a_i$'s are distinct. Now $f(\alpha,\alpha^2)=0$ gives
    \begin{equation*}
    \hspace{-40pt}(\alpha^{a_1+2 b_1}+\alpha^{a_2+2 b_2}+\alpha^{a_3+2 b_3}+\alpha^{a_4+2 b_4})(\alpha+\alpha^2+\alpha^5+\alpha^4)=0
    \end{equation*}
    This implies
    \begin{align*}
     & \alpha^{a_1+2 b_1}+\alpha^{a_2+2 b_2}+\alpha^{a_3+2 b_3}+\alpha^{a_4+2 b_4}=0 \\
     \implies &\alpha^{2 b_1}(\alpha^{a_1}+\alpha^{a_2})+\alpha^{2 b_3}(\alpha^{a_3}+\alpha^{a_4})=0\\
     \implies &(\alpha^{a_1}+\alpha^{a_2})(\alpha^{b_1}+\alpha^{b_3})^2=0,
     \end{align*}
     the last implication following from $f(\alpha,1)=\alpha^{a_1}+\alpha^{a_2}+\alpha^{a_3}+\alpha^{a_4}=0)$. Thus, either $a_1=a_2$ or $b_1=b_3$. If $a_1=a_2$, then we have $(a_1,b_1)=(a_2,b_2)$, which is a contradiction. If $b_1=b_3$, then we get $b_1=b_2=b_3=b_4$, and we have already assumed that $a_i$'s are not distinct, and so we again get a contradiction.
     \item[-] Assume that an odd number of the parameters $A$, $B$, $C$, $D$ is non-zero. Then, $f(1,1)=1 \neq 0$, so that $f(x,y) \notin C_U$. 
     \item[-] Assume that exactly two of $A$, $B$, $C$, $D$ are non-zero; say, $A=B=1$ and $C=D=0$. If $f(x,y) \in C_U$, then $f(\alpha,\alpha)=0$ implies $\alpha^{a_1+b_1}+\alpha^{a_2+b_2}=0$. This gives $a_1+b_1 \equiv a_2+b_2$ (mod $2^{4r}-1$). Similarly, $f(\alpha,\alpha^{-1})=0$ implies $\alpha^{a_1-b_1}+\alpha^{a_2-b_2}=0$. This gives $a_1-b_1 \equiv a_2-b_2$ (mod $2^{4r}-1$). 
    So, $2 a_1 \equiv 2 a_2$ (mod $2^{4r}-1$) which gives $a_1=a_2$. Thus $(a_1,b_1)=(a_2,b_2)$, which is a contradiction. 
    \item[-] Assume $(A,B,C,D)=(1,1,1,1)$. If $f(x,y) \in C_U$, then we have $f(\alpha,\alpha)=0$, which implies $\alpha^{a_1+b_1}+\alpha^{a_2+b_2}+\alpha^{a_3+b_3}+\alpha^{a_4+b_4}=0$. Proceeding similarly as in the case when  $(A,B,C,D)=(0,0,0,0)$, we will arrive at a contradiction. 
    \end{itemize}
\end{enumerate}
Thus, in all cases, we reach the conclusion that $f(x,y) \notin C_U$.
\end{proof}

\begin{claim}
        Polynomials of Type II are not in $C_U$.
        \label{claim:Type-II}
\end{claim}
\begin{proof}
Let $f(x,y)$ be a polynomial of Type~II. If $(A,B,C)=(0,0,0)$ or exactly two of $A$, $B$, $C$ are non-zero, then $f(1,1)=1 \neq 0$, so that $f(x,y) \notin C_U$. 

Assume that exactly one of $A$, $B$, $C$ is non-zero; say, $A=1$ and $B=C=0$. If $f(x,y) \in C_U$, then $f(\alpha,\alpha)=0$ and $f(\alpha,\alpha^{-1})=0$ gives $a_1+b_1+2 \equiv a_4+b_4$ (mod $2^{4r}-1$) and $a_1-b_1 \equiv a_4-b_4$ (mod $2^{4r}-1$). This implies $a_4 \equiv a_1+1$ (mod $2^{4r}-1$) and $b_4 \equiv b_1+1$ (mod $2^{4r}-1$). We split $a_4 \equiv a_1+1$ (mod $2^{4r}-1$) into two cases:
        \begin{itemize}
            \item[-] $a_4=a_1+1$. Then, $f(\alpha,1)=0$ gives 
            \begin{align*}
            \hspace{-20pt} \alpha^{a_1}(1+\alpha+\alpha^2)+\alpha^{a_2}(1+&\alpha^2)+\alpha^{a_3}(1+\alpha^2)+\alpha^{a_1+1}=0 \\
            \implies & (\alpha^{a_1}+\alpha^{a_2}+\alpha^{a_3})(1+\alpha^2) = 0 \\
            \implies & \alpha^{a_1}+\alpha^{a_2}+\alpha^{a_3} = 0.
            \end{align*}
            Similarly, from $f(\alpha^3,1)=0$, we get $\alpha^{3 a_1}+\alpha^{3 a_2}+\alpha^{3 a_3}=0$. This gives 
            \begin{align*}
                & \alpha^{3 a_1}+\alpha^{3 a_2}=\alpha^{3 a_3}=(\alpha^{a_1}+\alpha^{a_2})^3\\
                \implies & \alpha^{a_1+a_2}(\alpha^{a_1}+\alpha^{a_2})=0.
            \end{align*}
            This implies $a_1=a_2$. 
            Now, we bring in $b_4 \equiv b_1+1$ (mod $2^{4r}-1$). In either of the cases $b_4=b_1+1$ and $(b_1,b_4) = (2^{4r}-2,0)$, from $f(1, \alpha)=0$ and $f(1, \alpha^3)=0$, we can obtain $b_1=b_2$. Together with $a_1 = a_2$, we obtain $(a_1,b_1) = (a_2,b_2)$, which is a contradiction. 
            \item[-] $a_1=2^{4r}-2$ and $a_4=0$. An argument analogous to that in the case above will again lead to the contradiction $(a_1,b_1) = (a_2,b_2)$.
        \end{itemize}

 Finally, if $A=B=C=1$ then $f(1,\gamma)=\gamma^{b_4}\ne 0$, so that $f(x,y) \notin C_U$.
 \end{proof}

\begin{claim}
        Polynomials of Type III are not in $C_U$.
        \label{claim:Type-III}
\end{claim}
\begin{proof}
Let $f(x,y)$ be a polynomial of Type III.        
\begin{enumerate}
    \item Suppose that exactly one of $p_3$ and $p_4$ is non-zero; say, $p_3=1$ and $p_4 = 0$. Then, we have $f(x,y)=x^{a_1} y^{b_1}(x+y+Axy+x y^2+x^2 y)+x^{a_2} y^{b_2}(x+y+Bxy+x y^2+x^2 y)+x^{a_3} y^{b_3}$. If $(A,B)=(0,0)$ or $(A,B)=(1,1)$ then $f(1,1)=1 \neq 0$, which implies that $f(x,y) \notin C_U$. 

    \medskip
    
    So, let us assume that exactly one of $A$ and $B$ is non-zero; say, $A=1$ and $B=0$. If $f(x,y) \in C_U$, then $f(\alpha,\alpha)=0$ and $f(\alpha,\alpha^{-1})=0$, from which we obtain $a_1+b_1+2 \equiv a_3+b_3$ (mod $2^{4r}-1$) and $a_1-b_1 \equiv a_3-b_3$ (mod $2^{4r}-1$). This gives $a_3 \equiv a_1+1$ (mod $2^{4r}-1$) and $b_3 \equiv b_1+1$ (mod $2^{4r}-1$). Thus, $\bigl((a_3=a_1+1) \text{~or~} (a_1,a_3)=(2^{4r}-2,0)\bigr)$ and $\bigl((b_3=b_1+1) \text{~or~} (b_1,b_3)=(2^{4r}-2,0)\bigr)$
    \begin{itemize}
        \item[-] Consider $a_3=a_1+1$. In this case, $f(\alpha,1)=0$ gives
        \begin{align*}
            &\alpha^{a_1}(1+\alpha+\alpha^2)+\alpha^{a_2}(1+\alpha^2)+\alpha^{a_1+1}=0\\
            \implies & (\alpha^{a_1}+\alpha^{a_2})(1+\alpha^2)=0\\
            \implies & a_1=a_2.
        \end{align*}
        Now, either $(b_3=b_1+1) \text{~or~} (b_1,b_3)=(2^{4r}-2,0)$. In the first case, starting with $f(1,\alpha) = 0$, an analogous argument as above yields $b_1=b_2$. In the case when $(b_1,b_3)=(2^{4r}-2,0)$, $f(1,\alpha)=0$ gives $(\alpha^{-1}+\alpha^{b_2})(1+\alpha^2)=0$. This implies $b_2=2^{4r}-2$, so that $b_1 = b_2$ again. Thus, in either case, we end up with $(a_1,b_1) = (a_2,b_2)$, which is a contradiction.
        \item[-] The case $(a_1,a_3)=(2^{4r}-2,0)$ and $((b_3=b_1+1) \text{~or~} (b_1,b_3)=(2^{4r}-2,0))$ is handled similarly.
    \end{itemize}
    \item Suppose that $p_3=p_4=1$, so that $f(x,y)=x^{a_1} y^{b_1}(x+y+Axy+x y^2+x^2 y)+x^{a_2} y^{b_2}(x+y+Bxy+x y^2+x^2 y)+x^{a_3} y^{b_3}+x^{a_4} y^{b_4}$. Consider $(A,B)=(0,0)$. If $f(x,y) \in C_U$, then $f(\alpha,\alpha)=0$ and $f(\alpha,\alpha^{-1})=0$, from which we obtain $\alpha^{a_3+b_3}+\alpha^{a_4+b_4}=0$ and $\alpha^{a_3-b_3}+\alpha^{a_4-b_4}=0$. This implies $(a_3,a_4)=(b_3,b_4)$, a contradiction. 
    
    If exactly one of $A$, $B$ is non-zero, then $f(1,1)=1 \neq 0$, so that $f(x,y) \notin C_U$.
    
    Finally, consider $(A,B)=(1,1)$. If $f(x,y) \in C_U$, then $f(\alpha,\alpha^{-1})=0$ and $f(\alpha^{3},\alpha^{-3})=0$, from which we get
    \begin{align}
    \alpha^{a_1-b_1}+\alpha^{a_2-b_2}+\alpha^{a_3-b_3}+\alpha^{a_4-b_4}&=0 \label{eq:S2} \\
    \alpha^{3 a_1-3 b_1}+\alpha^{3 a_2-3 b_2}+\alpha^{3 a_3-3 b_3}+\alpha^{3 a_4-3 b_4}&=0 \label{eq:S3}.
    \end{align}
    
    Let $a_i-b_i\equiv c_i$ (mod $2^{4r}-1)$ for some $0 \leq c_i \leq 2^{4r}-2$, $i=1,2,3,4$. If the $c_i$'s are all distinct, then we are done, as there is no solution to Equations \eqref{eq:S2} and \eqref{eq:S3}. 
    
    So, consider the situation when the $c_i$'s are not all distinct, say $c_1=c_2$, i.e., $a_1-b_1 \equiv a_2-b_2$ (mod $2^{4r}-1)$. Then, from \eqref{eq:S2}, we also get $a_3-b_3 \equiv a_4-b_4$ (mod $2^{4r}-1)$. Now, $f(\alpha,1)=0$ and $f(1,\alpha)=0$ gives
    \begin{align*}
    (\alpha^{a_1}+\alpha^{a_2})(1+\alpha+\alpha^2)+\alpha^{a_3}+\alpha^{a_4}&=0 \\
    (\alpha^{b_1}+\alpha^{b_2})(1+\alpha+\alpha^2)+\alpha^{b_3}+\alpha^{b_4} &=0
    \end{align*}
    From these two equations, we obtain
    \begin{equation*}
        (\alpha^{a_3}+\alpha^{a_4})(\alpha^{b_1}+\alpha^{b_2})+(\alpha^{b_3}+\alpha^{b_4})(\alpha^{a_1}+\alpha^{a_2})=0
    \end{equation*}
    \begin{align*}
        \implies &(\alpha^{a_4-b_4+b_3}+\alpha^{a_4})(\alpha^{b_1}+\alpha^{b_2})\\
        &+(\alpha^{b_3}+\alpha^{b_4})(\alpha^{a_2-b_2+b_1}+\alpha^{a_2})=0\\
        \implies &\alpha^{a_4-b_4} (\alpha^{b_3}+\alpha^{b_4})(\alpha^{b_1}+\alpha^{b_2})\\
        &+(\alpha^{b_3}+\alpha^{b_4}) \alpha^{a_2-b_2}(\alpha^{b_1}+\alpha^{b_2})=0\\
        \implies & (\alpha^{a_2-b_2}+\alpha^{a_4-b_4})(\alpha^{b_1}+\alpha^{b_2})(\alpha^{b_3}+\alpha^{b_4})=0.
    \end{align*}
    Thus, either $b_1=b_2$ or $b_3=b_4$ or $a_2-b_2 \equiv a_4-b_4$ (mod $2^{4r}-1$).
    \begin{itemize}
        \item[-] If $b_1=b_2$ then we get $a_1=a_2$ as well, since $a_1-b_1 \equiv a_2-b_2$ (mod $2^{4r}-1)$. Thus, $(a_1,b_1)=(a_2,b_2)$, which is a contradiction.
        \item[-] If $b_3=b_4$ then we get $a_3=a_4$ as well, since $a_3-b_3 \equiv a_4-b_4$ (mod $2^{4r}-1)$. Thus, $(a_3,b_3)=(a_4,b_4)$, which is a contradiction.
        \item[-] Assume that $a_2-b_2 \equiv a_4-b_4$ (mod $2^{4r}-1$). Then, $f(\alpha,\alpha)=0$ gives
    \begin{equation*}
    \alpha^{a_1+b_1+2}+\alpha^{a_2+b_2+2}+\alpha^{a_3+b_3}+\alpha^{a_4+b_4}=0
    \end{equation*}
    \begin{equation*}
    \begin{split}
        \implies \alpha^{a_2-b_2+2 b_1+2}+\alpha^{a_2+b_2+2} &+\alpha^{a_4-b_4+2 b_3}\\
        &+ \alpha^{a_4+b_4}=0
        \end{split}
    \end{equation*}
    \begin{equation*}
        \hspace{-15pt}\implies \alpha^{a_2-b_2+2} (\alpha^{2 b_1}+\alpha^{2 b_2})+\alpha^{a_2-b_2} (\alpha^{2 b_3} +\alpha^{2 b_4})=0.
    \end{equation*}
    This gives $\alpha(\alpha^{b_1}+\alpha^{b_2})=\alpha^{b_3}+\alpha^{b_4}$. Using $f(1, \alpha)=0$ now gives $(\alpha^{b_1}+\alpha^{b_2})(1+\alpha^2)=0$. This implies $b_1=b_2$ and hence, from $a_1-b_1 \equiv a_2-b_2$ (mod $2^{4r}-1)$, we again get $(a_1,b_1)=(a_2,b_2)$, a contradiction.
    \end{itemize}
\end{enumerate}
Thus, in all cases, we reach the conclusion that $f(x,y) \notin C_U$.
\end{proof}

\begin{claim}
    Polynomials of Type IV are not in $C_U$.
    \label{claim:Type-IV}
\end{claim}   
\begin{proof}
Let $f(x,y)$ be a polynomial of Type IV.        
\begin{enumerate}
    \item Suppose that exactly one of $p_2$, $p_3$, $p_4$ is non-zero; say, $p_2=1$ and $p_3=p_4=0$. Then, we have $f(x,y)=x^{a_1} y^{b_1}(x+y+Axy+x y^2+x^2 y)+x^{a_2} y^{b_2}$. In the case when $A=0$, if $f(x,y) \in C_U$, then $f(\alpha,\alpha)=0$, which implies $\alpha^{a_2+b_2}=0$, which is not possible. When $A=1$, then $f(1,\gamma)=\gamma^{b_2} \ne 0$, so that $f(x,y) \notin C_U$.
    \item Suppose that exactly two of $p_2$, $p_3$, $p_4$ are non-zero; say, $p_2=p_3=1$ and $p_4=0$. Then, $f(x,y)=x^{a_1} y^{b_1}(x+y+Axy+x y^2+x^2 y)+x^{a_2} y^{b_2}+x^{a_3} y^{b_3}$. When $A=0$, if $f(x,y) \in C_U$, then $f(\alpha,\alpha)=0$ and $f(\alpha,\alpha^{-1})=0$ gives $a_2+b_2 \equiv a_3+b_3$ (mod $(2^{4r}-1)$) and $a_2-b_2 \equiv a_3-b_3$ (mod $(2^{4r}-1)$). From this, we obtain $(a_2,b_2)=(a_3,b_3)$, a contradiction. When $A=1$, then $f(1,1)=1 \neq 0$, so that $f(x,y) \notin C_U$.
    \item Suppose that $p_2=p_3=p_4=1$. Then, $f(x,y)=x^{a_1} y^{b_1}(x+y+Axy+x y^2+x^2 y)+x^{a_2} y^{b_2}+ x^{a_3} y^{b_3}+ x^{a_4}y^{b_4}$. If $A=0$, then $f(1,1)=1 \neq 0$, which implies that $f(x,y) \notin C_U$. 
            
    So, let us consider $A=1$. If $f(x,y) \in C_U$, then $f(\beta,\beta^2)=0$ implies 
    $$\beta^{a_2+2 b_2}+\beta^{a_3+2 b_3}+\beta^{a_4+2 b_4}=0.$$
    Let $a_i+2 b_i \equiv c_i$ (mod $5$) for $i=2,3,4$. If all $c_i$'s are not distinct, say $c_2=c_3$, then we get $\beta^{c_4}=0$, which is impossible. So the $c_i$'s are all distinct, so that we have $\beta^{c_2}+\beta^{c_3}+\beta^{c_4}=0$. Let us assume WLOG that $c_2<c_3<c_4$ then we get $1+\beta^{c_3-c_2}+\beta^{c_4-c_2}=0$ where $1 \leq c_3-c_2, c_4-c_2 \leq 4$ are distinct. But $\beta$, being a fifth root of unity, does not satisfy such an equation, again leading to a contradiction.
\end{enumerate}
Thus, in all cases, we reach the conclusion that $f(x,y) \notin C_U$. 
\end{proof}

\section{ Proof of Proposition \ref{UZU_prop}} \label{UZU proof}

Suppose that $U$ is a Pauli operator and $\mathcal{S}$ is the stabilizer group of a graph state $\ket{\mathrm{G}}$. We wish to prove that for any non-empty subset of indices $A$,
\begin{equation*}
    U Z_A \in \mathcal{S} \implies U \notin \mathcal{S}.
\end{equation*}

We first note that $Z_A \notin \mathcal{S}$, since it anti-commutes with the stabilizer generator $X_i Z_{N(i)}$ for any $i \in A$. Hence, $\braket{\mathrm{G} | Z_A | \mathrm{G}} = 0$.

Now, suppose that $U Z_A \in \mathcal{S}$. Then,
\begin{equation*}
    U Z_A = S_{i_1} \cdot S_{i_2} \cdot S_{i_3} \cdots\  S_{i_r},
\end{equation*}
for some choice of $r$ stabilizer generators of $\mathcal{S}$.
Right multiplication with $Z_A$ on both sides yields
\begin{equation*}
     U = S_{i_1} \cdot S_{i_2} \cdot S_{i_3} \cdots\ S_{i_r} \cdot Z_A
\end{equation*}
We then have
\begin{align*}
     \braket{\mathrm{G} | U | \mathrm{G}} &= \braket{\mathrm{G} | S_{i_1} \cdot S_{i_2} \cdot S_{i_3} \cdots\ S_{i_r} \cdot  Z_A | \mathrm{G}} \\ 
     &= \braket{\mathrm{G} | Z_A | \mathrm{G}} \ = \ 0.
\end{align*}
Therefore $U\notin \mathcal{S}$.

\newpage
\end{document}